\documentclass[12pt, reqno, a4paper,oneside]{amsart}

\rmfamily{\fontsize{12pt}{\baselineskip}\selectfont}

\usepackage{graphicx}
\usepackage{amsthm}
\usepackage{mathrsfs, amsmath, amssymb, braket, bm, mathtools,enumitem,color,marginnote, esint}
\usepackage[bbgreekl]{mathbbol}
\usepackage{amsfonts}
\usepackage{bm}
\usepackage{color}
\usepackage{xcolor}
\usepackage{hyperref}
\usepackage{mathrsfs}
\usepackage{longtable}
\usepackage[top=2.5cm,bottom=2.0cm,left=2.0cm,right=2.0cm]{geometry}
\usepackage{epstopdf}
\usepackage{stmaryrd}
\usepackage[normalem]{ulem}
\usepackage{cancel}
\usepackage{scalerel}
\usepackage{enumitem}
\usepackage{algorithm}
\usepackage{subfigure}
\usepackage{caption}
\usepackage{algorithm, algorithmicx}
\usepackage[noend]{algpseudocode}
\usepackage{lineno}
\usepackage{enumitem}
\usepackage{diagbox}
\usepackage{multirow}
\usepackage{cleveref}

\usepackage{environments}
\numberwithin{theorem}{section}
\numberwithin{equation}{section}

\renewcommand{\algorithmiccomment}[1]{\bgroup\hfill//~#1\egroup}

\newcommand{\Ham}{{\mathcal H}}
\newcommand{\ep}{\varepsilon}
\newcommand{\R}{\mathbb{R}}
\newcommand{\C}{\mathbb{C}}
\newcommand{\Z}{\mathbb{Z}}

\newcommand{\dd}{{\,{\rm d}}}
\newcommand{\m}{\mathfrak{m}}

\newcommand{\Gap}{\mathsf{gap}_+}
\newcommand{\gap}{\mathsf{gap}_-}
\newcommand{\tr}{\mathrm{Tr}}
\newcommand{\A}{\mathsf{A}}

\newcommand{\cc}{\mathscr{C}}
\newcommand{\dist}{\mathrm{dist}}

\newcommand{\adm}{\mathrm{Adm}}

\newcommand{\br}{\bm{r}}
\newcommand{\bxi}{\bm{\xi}}
\newcommand{\bzeta}{\bm{\zeta}}

\newcommand{\D}{\mathfrak{d}}
\newcommand{\Dr}{\mathfrak{d}^{\mathrm{ref}}}
\newcommand{\nob}{N_{\mathrm{b}}}
\newcommand{\na}{M}
\newcommand{\eF}{\ep_{\mathrm{F}}}

\newcommand{\CT}{\gamma_{\rm{CT}}}
\newcommand{\loc}{\rm{loc}}
\newcommand{\dmr}{\rho^{\mathrm{ref}}}
\newcommand{\dmu}{\rho(u)}
\newcommand{\dmuEp}{\rho(u_{\ep})}
\newcommand{\Hr}{\Ham^{\mathrm{ref}}}

\newcommand{\Hxi}{\Ham_{\pmb{\xi}}}
\newcommand{\etaP}{\eta_{+}}
\newcommand{\etaM}{\eta_{-}}

\newcommand{\ccRef}{c_{\mathrm{r}}}
\newcommand{\reg}{\delta_{*}}

\def\asTB{\textnormal{\bf (TB)}}
\newcommand{\asGap}{\textup{\textbf{(Gap)}}}
\newcommand{\asReg}{\textup{\textbf{(Reg)}}}

\DeclareMathOperator{\re}{Re}

\title{Nearsightedness in Materials with Indirect Band Gaps}

\author{Juerong Feng}
\address{
School of Mathematical Sciences, Beijing Normal University, China.
}
\email{juerong\_f@mail.bnu.edu.cn} 

\author{Huajie Chen}
\address{
School of Mathematical Sciences, Beijing Normal University, China.
}
\email{chen.huajie@bnu.edu.cn}

\author{Christoph Ortner}
\address{
University of British Columbia, 1984, Mathematics Road, Vancouver, BC, Canada 
}
\email{ortner@math.ubc.ca}

\author{Jack Thomas}
\address{
School of Mathematics, University of Minnesota Twin Cities, Minneapolis, MN
 55455, USA
}
\email{thom9218@umn.edu}

\date{\today}

\begin{document}

\begin{abstract}
 We investigate the nearsightedness property in the linear tight binding model at zero Fermi-temperature. We focus on the decay property of the density matrix for materials with indirect band gaps. By representing the density matrix in reciprocal space, we establish a qualitatively sharp estimate for the exponential decay rate in homogeneous systems, possibly with localized perturbations.
 This work refines the estimates presented in  (Ortner, Thomas \& Chen, 2020) for systems with small band gaps. 
\end{abstract}

\maketitle
\thispagestyle{empty}

\section{Introduction}
\label{sec:Introduction}
The electronic structure determines a wide range of physical and chemical properties of materials. For its computation, density functional theory (DFT) has been successfully applied for many decades \cite{finnis2003interatomic, martin2020electronic}. However, even in the minimalistic tight binding models, practical implementations typically incur a high computational cost, scaling cubically with the system size. One approach to alleviate this high cost is through linear scaling algorithms \cite{bowler2012methods, goedecker1999linear, kohn1996density, niklasson2011density}, which rely on a notion of locality of electronic systems, often known as ``nearsightedness”. 

First introduced by Kohn~\cite{kohn1996density, prodan2005nearsightedness}, nearsightedness describes the principle that interactions in many-atom systems are predominantly localised. Mathematically, this corresponds to the exponential off-diagonal decay of the density matrix. The decay rate determines the error committed when truncating it during a numerical simulation, constituting the foundation for developing linear scaling techniques in electronic structure calculations.

Therefore, exploring the dependence of the decay rate on various factors is important and has provoked extensive discussion, both analytically and numerically \cite{benzi2013decay, chen2016qm, goedecker1998decay, jewski2004exact,  ortner2020locality}: it is determined by the band gap in insulators and by the Fermi-temperature in metals. 
However, these results may not always be sufficiently precise to account for the exponential decay observed in some materials with small band gaps, e.g. in semiconductors like silicon and germanium. Numerical results from \cite{ortner2020locality} show that in the silicon system, the decay is nearly identical to the carbon system even though the former has a much smaller band gap.

The main purpose of this paper is to explore a subtle point that, to the best of our knowledge, has been missed in the nearsightedness discussion to date:  the decay rate of the density matrix depends not on the {\em gap} (the most common notion of a band gap) but on the {\em vertical gap}; cf.~Figures~\ref{fig:gap_type} and~\ref{fig:gap_bz}.

Our main results establishes this, for the case of linear tight binding models at zero Fermi-temperature, and show that this result is stable under perturbations.

This result is important because the vertical gap can be much larger than the gap, especially in small-gap semiconductors such as ${\rm Mg}_2{\rm Si}$; cf. Figure~\ref{fig:gap_bz}. The result explains why some small gapped systems still exhibit excellent exponential localization and linear scaling methods remain applicable. 

Next, we explore whether such improved locality estimates also generalize to interatomic forces, which is called {\em strong locality} in \cite{Cances2025, ortner2020point}.
Through a combination of analytical and numerical empirical results we show that this is not the case. This negative result demonstrates that locality of interatomic forces is a fundamentally stronger concept than classical nearsightedness.

\subsection{Outline}
\label{subsec:outline}
In \Cref{sec:tight_binding_models}, we introduce the physical background and linear tight binding models. We then provide definitions for the gap and vertical gap in reciprocal space.
In \Cref{sec:results}, we present the main results of this work. We begin by discussing the dependence of nearsightedness principle in homogeneous systems, and then extend the analysis to strong locality by investigating the decay rate for the derivatives of the density matrix. Under additional regularity assumption, we further provide sharp estimates for both cases.
Finally, we 
extend our analysis to inhomogeneous systems with bounded perturbations. 
In \Cref{sec:numerical}, 
we provide numerical experiments to verify our analytical results. The main conclusions of this paper are summarised in \Cref{sec:Conclusions}. 
All proofs are collected in \Cref{sec:Proofs of the Main Results}. 
In the appendix, we provide details of the Bloch transform, along with the specifics of numerical experiments and estimates for the exponential sum.

\subsection{Notation}
\label{subsec:Notations}
The symbol $\|\cdot\|$ and $|\cdot|$ will denote the $\ell^2$ and Euclidean norms on $\R$ or $\C$, respectively. The subscript of $\|\cdot\|$ indicates the norm is taken in which sense. In particular, we will use $\|\cdot\|_{\rm F}$ to denote the Frobenius norm for matrices. The ball of radius $R$ about $\ell$ and ${\bf 0}$ will be denoted by $B_{R}(\ell)$ and $B_R$, respectively. For a shifted lattice,
we will write $\Lambda - \ell := \{x - \ell: x \in \Lambda \}$. For two sets $A$ and $B$, we denote the Minkowski sum of them by $A + B := \{ \bm{a} + \bm{b} ~|~ \bm{a}\in A, \bm{b} \in B \}$, which is formed by adding each vector in $A$ and $B$. For a self-adjoint operator $T$, the spectrum of $T$ will be denoted by $\sigma(T)$.

We will use $\oint$ to denote contour integral and $\dashint$ to denote the average integral over its domain, that is,
\[
\dashint_{\Omega} := \frac{1}{|\Omega|}\int_{\Omega},
\]
where $|\Omega|$ is the volume of the integral domain $\Omega$. 

The symbol $C$ is a positive generic constant that may vary between successive lines of estimation. When estimating rates of decay, $C$ consistently remains independent of the system size and lattice position. The dependencies of $C$ will be typically evident from the context or explicitly stated. For simplicity, we sometimes write $f \lesssim g$ to mean $f \leq Cg$ for certain generic positive constant as above.

\section{Tight Binding Models}
\label{sec:tight_binding_models}

\subsection{Atomistic Model and Admissible Configurations}
\label{sec:config}
We begin with a homogeneous multi-lattice $\Lambda$ defined as follows: Let $\A \Z^d$ be a Bravais lattice (where $\A$ is non-singular and $d$ is the dimension of the system) and $\{p_i\}_{i=1}^{\na}\subset \R^d$ be a set of $\na$ shift vectors. 
Then, we define
\begin{align}
    \label{eq:lattice}
    \Lambda := \bigcup_{i = 1}^{\na} \big(p_i + \A \Z^d\big)
\end{align}
which is formed by taking the union of shifted Bravais lattice $\A\Z^d$ along $p_i$, $i = 1, \dots, \na$. We denote by $\Gamma$ a unit cell of $\A \Z^d$ including $\na$ atoms, and by $\Gamma^*$ the corresponding unit cell of the reciprocal lattice $2\pi\A^{-\mathrm{T}}\Z^d$. For example, we can simply take $\Gamma:= \A[-\frac{1}{2}, \frac{1}{2})^d $ and $\Gamma^*:= 2\pi \A^{-\mathrm{T}}[-\frac{1}{2}, \frac{1}{2})^d$, but other choices are possible.

Following \cite{chen2016qm,  ortner2020point, ortner2020locality}, we consider displacements $u: \Lambda \rightarrow \R^d$ and thus configurations of the form $\big( \ell + u(\ell) \big)_{\ell \in \Lambda}$. For brevity, we write 
\begin{align}
    \br_{\ell k }(u):=\ell + u(\ell) - k - u(k) 
    \quad \text{and} \quad
    r_{\ell k}(u) := |\br_{\ell k }(u)|.
\end{align}
When $u(\ell) =0$ for all $\ell \in \Lambda$, we will omit $u$ and take $\br_{\ell k }$, $r_{\ell k}$ for the remainder of the paper. 
Let us define the finite difference of $u$ (the strain) by
\begin{align}
D_{\rho} u(\ell)\coloneqq u(\ell + \rho) - u(\ell), \quad \forall \ell \in \Lambda,~ \rho \in \Lambda - \ell.
\end{align}
The infinite difference stencil is defined to be $Du(\ell)\coloneqq \{D_{\rho}u(\ell)\}_{\rho\in\Lambda - \ell}$, endowed with the norm 
\begin{align}
\label{eq:l2_norm}
    \|Du\|_{\ell_{\Upsilon}^2}:=\left(\sum_{\ell \in \Lambda} \left|Du(\ell)\right|^2_{\Upsilon}\right)^{1 / 2}
    ~~\text{with}~~
    \left|Du(\ell)\right|^2_{\Upsilon} :=
    \sum_{\rho \in \Lambda-\ell} e^{-2 \Upsilon|\rho|}\left|D_\rho u(\ell)\right|^2.
\end{align}
Since all of the semi-norms $\|D\cdot\|_{\ell^2_{\Upsilon}}$ for $\Upsilon>0$ are equivalent~\cite{chen2019geometry}, we will fix $\Upsilon>0$ for the remainder. 
We wish to consider the set of finite energy displacements with the following non-interpenetration condition:
\begin{align}
\label{eq:adm_configs}
    \adm(\Lambda):=\left\{u: \Lambda \rightarrow \R^d \colon  \|D u\|_{\ell_{\Upsilon}^2}<\infty, \,\,  \exists \,  \m > 0 ~\text{with}~ r_{\ell k}(u)\geq \m 
    \,\, \forall \ell, k \in \Lambda
    \right\}.
\end{align}

\subsection{Tight Binding Model}
\label{sec:TB}
Let $\nob$ denote the number of atomic orbitals per atom, indexed by $1\leq a, b \leq \nob$. To each atomic site $\ell \in \Lambda$, we denote the atomic species by $z_\ell$. For a given admissible configuration $u \in \mathrm{Adm}(\Lambda)$, the state for the bond $(\ell, k) \in \Lambda \times \Lambda$ is defined by ${\bm x}_{\ell k}(u) = (\br_{\ell k}(u), z_\ell, z_k)$. For a function $f({\bm x}_{\ell k})$ we write $\nabla f({\bm x}_{\ell k}) := \nabla_{\br_{\ell k}} f( (\br_{\ell k}, z_\ell, z_k) )$. 

The linear tight binding Hamiltonian is given by the following:

\vspace{.5em}

\begin{flushleft}
    $\asTB$ ~
    Suppose that the Hamiltonian takes the form of 
    \begin{align}
    \label{eq:ham}
    \Ham(u)_{\ell k }^{ab} = h^{ab}\big({\bm x}_{\ell k}(u)\big) + \delta_{\ell k} \sum_{m \neq \ell} t^{ab}({\bm x}_{\ell m}(u))
    \end{align} 
for $\ell, k \in \Lambda$ and $1\leq a, b \leq \nob$. 
The functions $h^{ab}, t^{ab}: \R^d \times \{z_{\ell} \}_{\ell \in \Lambda} \times \{z_{\ell}\}_{\ell \in \Lambda}\rightarrow \R$ are continuously differentiable (in the first variable) with
\begin{align}
\label{eq:ham_decay}
    & \left|h^{ab}({\bm x}_{\ell k})\right| +
    \left|t^{ab}({\bm x}_{\ell k})\right| 
    +\left|\nabla h^{ab}({\bm x}_{\ell k})\right| 
    + \left|\nabla t^{ab}({\bm x}_{\ell k})\right|
    \leq h_0 e^{-\gamma_0r_{\ell k}} 
\end{align}
for some $h_0, \gamma_0 > 0$. 
Moreover, we assume that $h^{ab}\big((\bm r_{\ell k}(u), z_\ell, z_k)\big) = h^{ba}\big((-{\bm r}_{\ell k}(u), z_k, z_\ell)\big)$.
\end{flushleft}

\vspace{.5em}

In the following, we will denote by $\Ham^{\mathrm{ref}} \coloneqq \Ham(\bm 0)$ the reference Hamiltonian.
 
\begin{remark}
    We note that the constants $h_0, \gamma_0$ in \eqref{eq:ham_decay} may be chosen to be independent of atomic species.
    Additionally, \eqref{eq:ham_decay} requires that the dependence of $\Ham(u)^{ab}_{\ell k}$ on site $m$ decays exponentially in $r_{\ell m}(u) + r_{m k}(u)$. That is, we are assuming long-range Coulomb interactions are screened, a common assumption in many practical tight binding models~\cite{cohen1994tight, mehl1996applications, papaconstantopoulos1997tight}. We assume this to separate Coulomb type long-range effects from those arising due to small band gaps.
    
    Under the assumption $\mathbf{(TB)}$, we have $\sigma\big(\Ham(u)\big) \subset [\underline{\sigma}, \overline{\sigma}]$, where $\underline{\sigma}$,  $\overline{\sigma}$ are real and only depend on $d, \m, h_0, \gamma_0$, while independent of the system size. The detailed proof leveraging the Gershgorin circle theorem has been given in \cite{chen2016qm}. 
    
    The number of orbitals, $\nob$,  generally depends on the atom species and thus on the atom sites. Only for the sake of simplicity of notation we assume this number is the same for all species. This assumption is easily avoided; see \cite{ortner2020locality} for a possible approach.
\end{remark} 

\subsection{Direct and Indirect Band Gaps}
\label{sec:Gap} 
Next, we review the concepts of direct and indirect band gaps. For this purpose, we first recall the Bloch transform, which provides a decomposition of $\Ham^{\rm ref}$, denoted by $\Hxi$ for all $\bxi \in \Gamma^*$. For more details, see \Cref{appd:bloch_trans}. By applying an eigenvalue decomposition for $\Hxi$, we can obtain energy band structure of the system and subsequently distinguish between direct and indirect band gaps. The Bloch fibres of $\Ham^{\rm ref}$ are given by

    \begin{align}
    \label{eq:block_ham}
    \big[\Hxi\big]^{ab}_{\ell_0 k_0} \coloneqq \sum_{\alpha\in \mathbb Z^d} [\Hr]^{ab}_{\ell_0 + \mathsf{A}\alpha, k_0} e^{- i (\ell_0 - k_0 +\mathsf{A}\alpha) \cdot \bxi} \quad \forall \bxi \in \Gamma^*,
    \end{align}
for all $\ell_0, k_0 \in \Lambda_0 \coloneqq \Gamma \cap \Lambda$ and $1\leq a, b \leq \nob$.
 The size of the matrix $\Hxi$ is $\na\nob\times \na\nob$. For $\ell \in \Lambda$, we write $\ell_0$ for the unique site in $\Lambda_0$ satisfying $\ell -\ell_0 \in \mathsf{A}\mathbb Z^d$.
Denoting the eigenvalues of $\Hxi$ by $\ep_1(\bxi) \leq \cdots \leq \ep_{\na\nob}(\bxi)$, we have
\begin{align}
    \label{eq:spectrum}
         \sigma(\Hr) = \bigcup_{j=1}^{\na\nob}\left\{\ep_j(\bxi): ~\bxi \in \Gamma^*\right\},
    \end{align}
where $\ep_j: \Gamma^* \to \R$ are referred to as the energy bands \cite{kittel2018introduction}. 
    
We focus exclusively on systems with a positive band gap, such as insulators and semi-conductors, where the Fermi level, $\eF$, lies inside the band gap:

\vspace{.5em}

\begin{flushleft}
    $\asGap$ ~  The Fermi level $\eF$ separates the energy bands into the valence bands and conduction bands \cite{cox1987electronic}:
    \begin{align}
    \eF \not \in \sigma(\Hr)  \quad\text{and}\quad \ep_{N_0}(\bxi) < \eF < \ep_{N_0+1}(\bxi) \quad\forall \bxi \in \Gamma^*,
    \end{align}
    for some $N_0$, the total number of occupied states in the system.
\end{flushleft}

\vspace{.5em}

In semiconductors, one can consider two types of band gaps $\gap \leq \Gap$ defined as 
\begin{align}
\label{eq:Gap}
\Gap &\coloneqq \min_{\bxi \in \Gamma^*} \big( 
\ep_{N_0+1}(\bxi) - \ep_{N_0}(\bxi)
\big), \qquad \text{and} \\ 
\label{eq:gap}
    \gap
     &\coloneqq \min_{\bxi \in \Gamma^*} \ep_{N_0+1}(\bxi) - \max_{\bxi \in \Gamma^*} \ep_{N_0}(\bxi).  
\end{align}

\begin{figure}[H]
    \centering
    \subfigure{
    \includegraphics[width = 7.6cm]{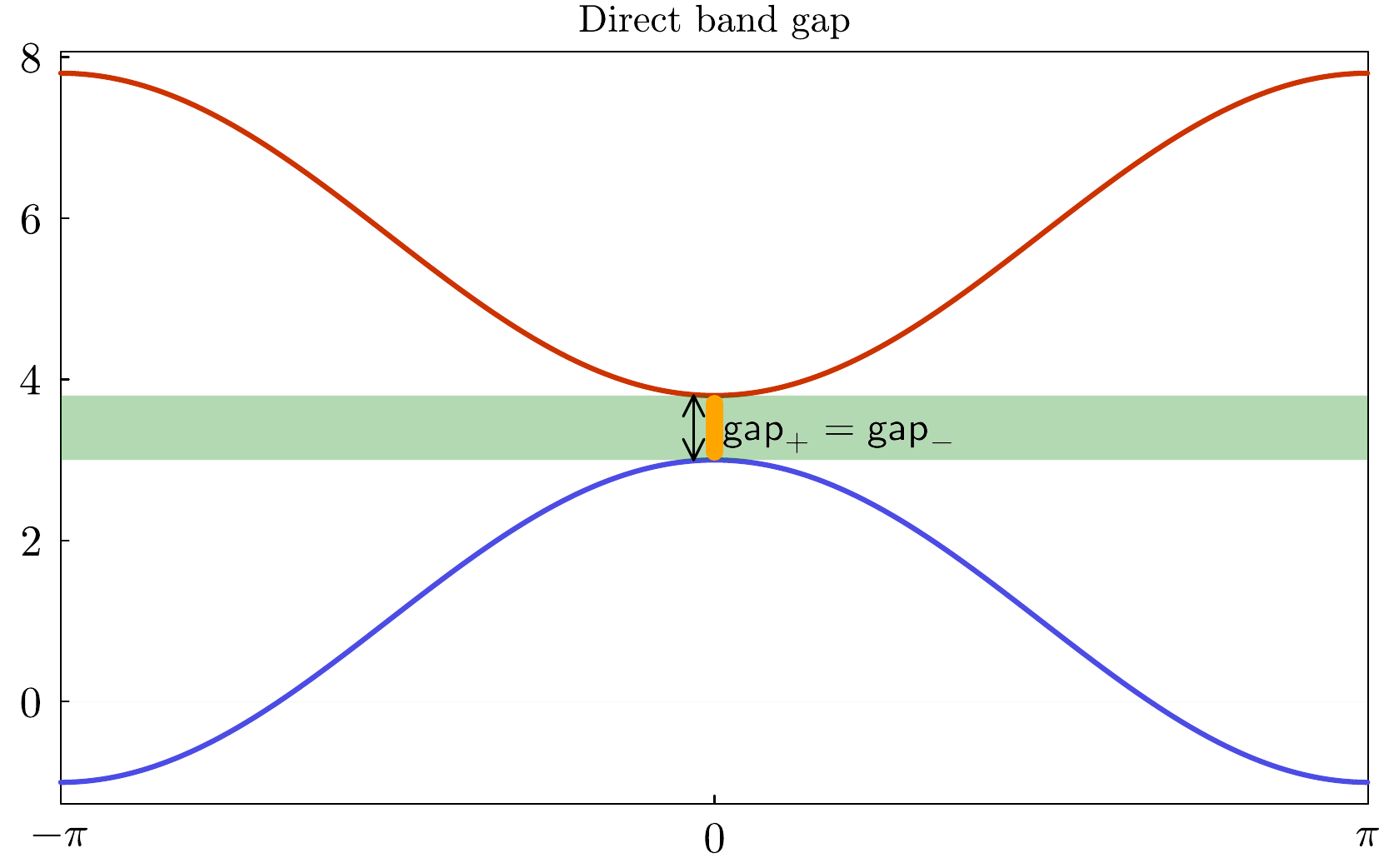}
    }
    \subfigure{
    \includegraphics[width = 7.6cm]{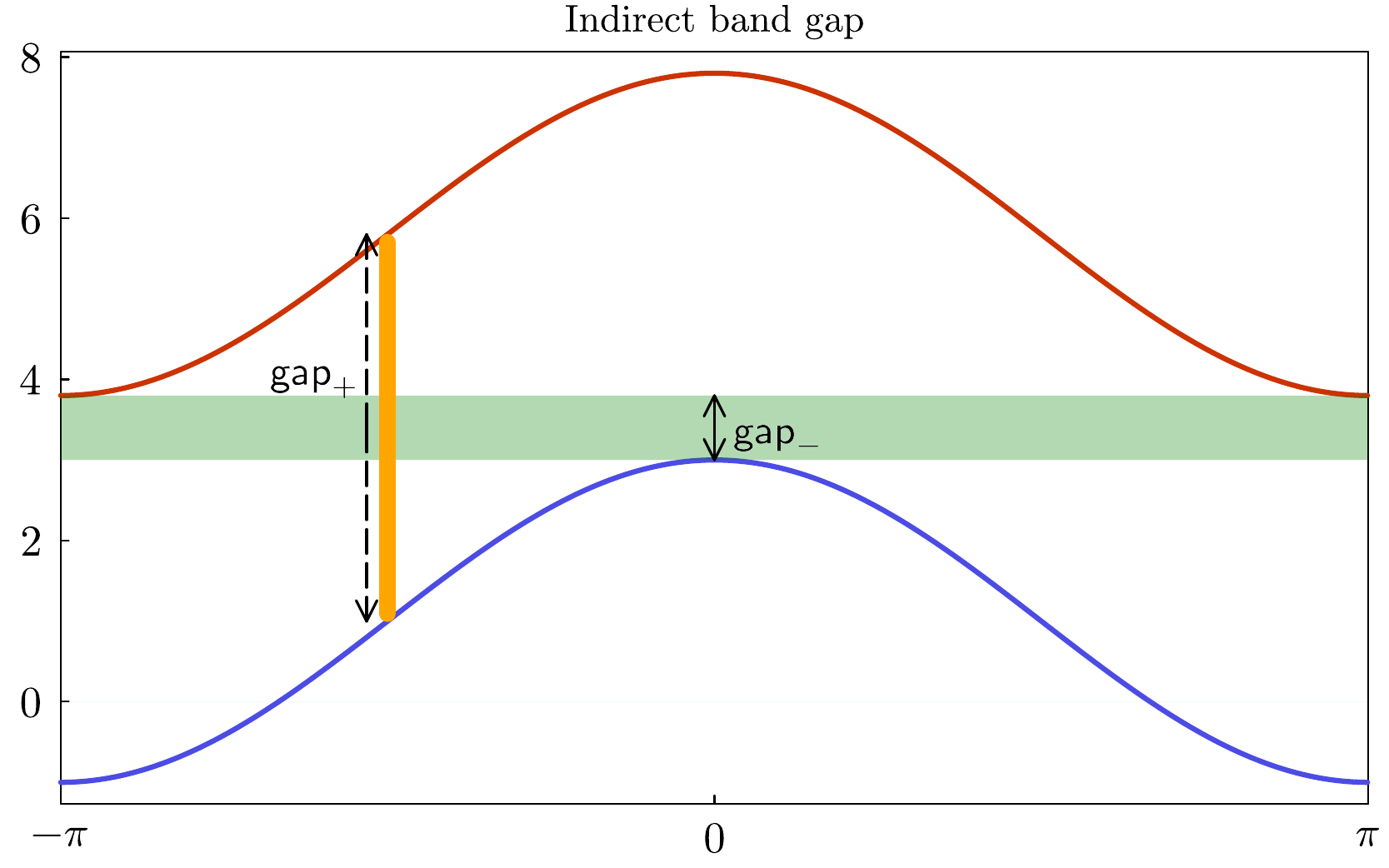}
    }
    \caption{Schematic of the direct and indirect band gaps.}
    \label{fig:gap_type}
\end{figure}
\begin{figure}[H]
    \centering
    \subfigure{
    \includegraphics[width = 7.8cm]{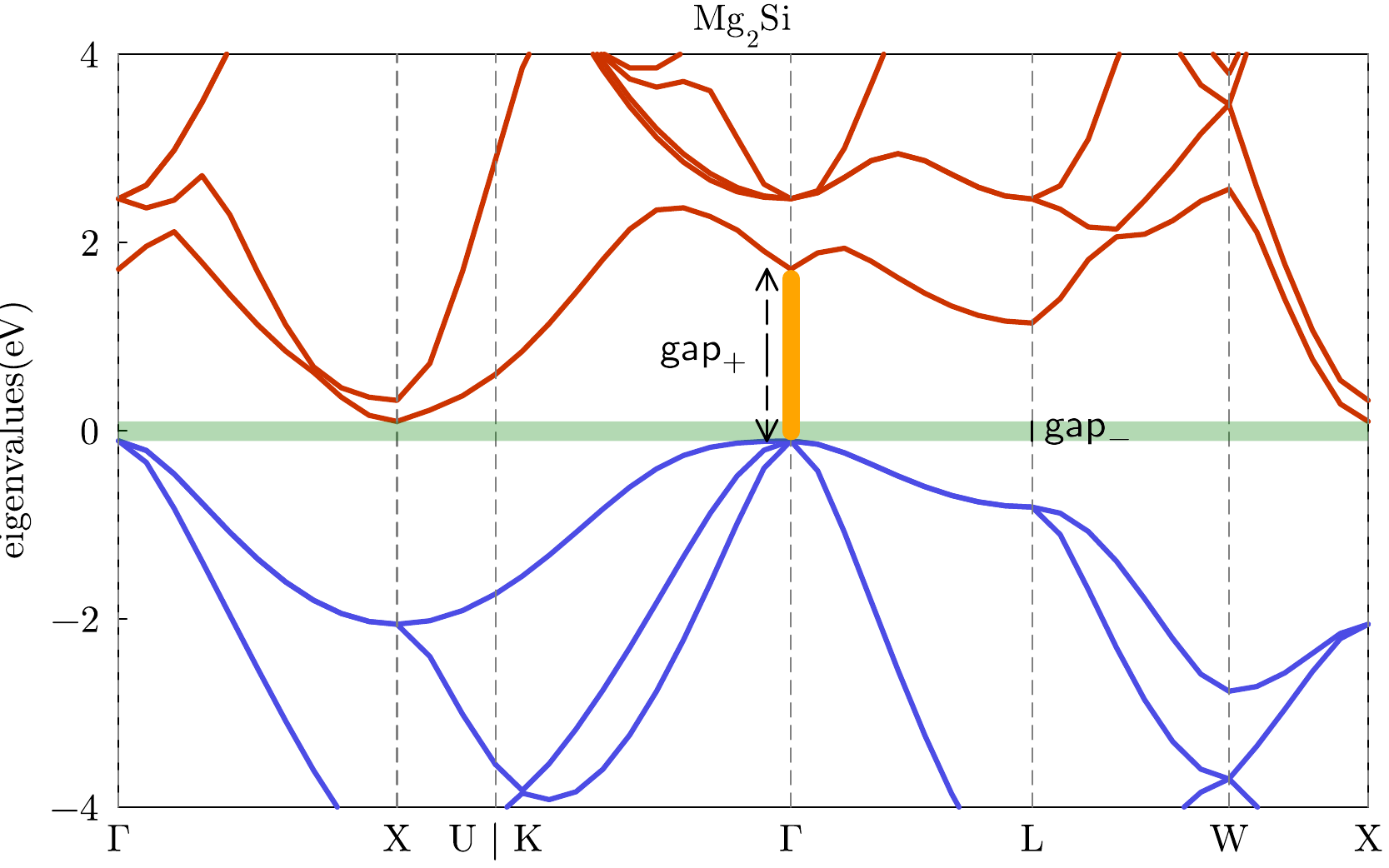}
    }
    \subfigure{
    \includegraphics[width = 7.8cm]{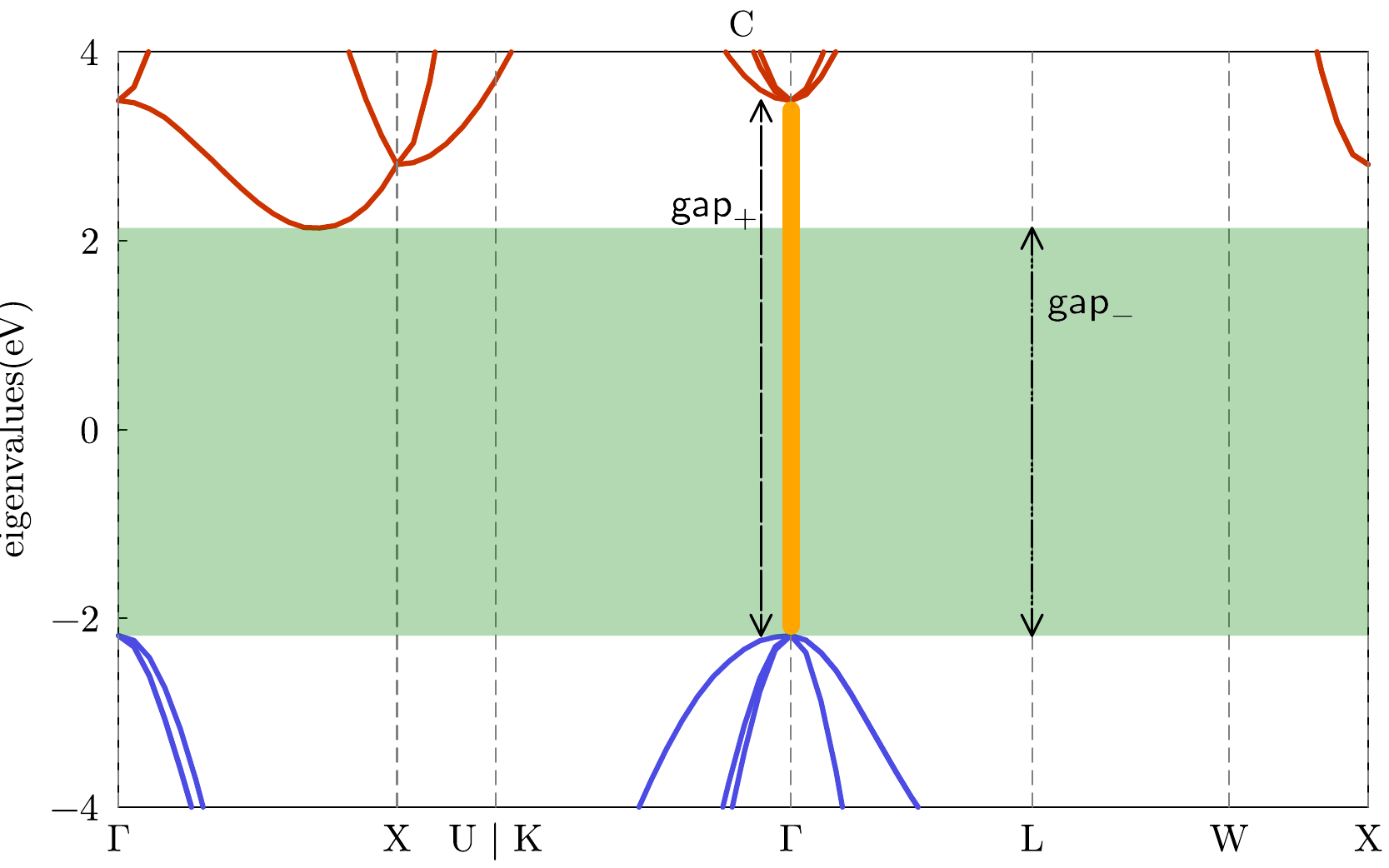}
    }
    \caption{The band structure of cubic $\rm{Mg_2Si}$ and cubic diamond carbon. In diamond, we have $\gap \approx \Gap$, whereas in $\rm{Mg_2Si}$, we observe $\gap \ll \Gap$.}
    \label{fig:gap_bz}
\end{figure}

Notice that when $\gap = \Gap$, the maximal energy of the valence bands and the minimal energy of the conduction bands occurs at the same vector $\bxi \in \Gamma^*$,  and the band gap will be referred to as a \textit{direct band gap}. Conversely, if $\Gap > \gap$, the two extrema occur at different $\bxi$ vectors, resulting in an \textit{indirect band gap}. We exhibit the disparity by two schematic bands in the \Cref{fig:gap_type}, where we use the green strip and the orange bar to denote $\gap$ and $\Gap$, respectively.
In \Cref{fig:gap_bz}, we present the band structure of $\rm{Mg_2S}i$ and $\rm{C}$, with indirect band gaps. The calculation is performed using the $\mathsf{DFTK.jl}$ package \cite{herbst2021dftk}.

In physics, the disparity of direct and indirect band gaps leads to various properties of semiconductors, which are associated with electron-hole pair generation and certain optical effects.
These effects depend on electronic transitions between the top of the valence bands and the bottom of the conduction bands \cite{cox1987electronic, kittel2018introduction, pankove1975optical}. Here, we investigate the role of the indirect band gap in the locality of electronic structure.

\section{Main Results}
\label{sec:results}
In this section, we present our main results on the \textit{nearsightedness} property,
which we sometimes call {\em weak locality} to distinguish it from locality of the mechanical response ({\em strong locality}). 

At zero Fermi-temperature, the Fermi--Dirac distribution is given by $\chi_{(-\infty, \eF)}$ and thus the density matrix reads 
 \begin{align}
     \dmu = \chi_{(-\infty, \eF)}(\Ham(u)), \quad u \in \adm(\Lambda).
 \end{align}
For simplicity, we will denote $\dmr := \rho(\bm 0)$, and use $\left|\dmr_{\ell k}\right|$ to refer to the Frobenius norm of each block for any $\ell, k \in \Lambda$.

We begin by considering the homogeneous lattice and present an improved estimate for the locality of the density matrix, explicitly capturing its linear dependence on $\Gap$: 

\begin{theorem}[Nearsightedness]
\label{thm:nearsight}
    Suppose that $\Lambda$ has the form of \eqref{eq:lattice} and $\Hr $ satisfies assumptions $\asTB$ and $\asGap$. Then, there exist constants $C_1, \eta_+ > 0$ such that 
    \begin{align}
    \label{eq:thm_nearineq}
        \left|\dmr_{\ell k}\right| \leq \frac{C_1}{\mathsf{gap}_+} e^{-{\etaP}r_{\ell k }}, \quad \forall \ell, k \in \Lambda,
    \end{align}  
    where $\etaP := c_1 \frac{\gamma_0}{h_0} \min\left\{h_0,\gamma_0^{d} \, \mathsf{gap}_+\right\}$ and $c_1, C_1 > 0$ only depend on $d, \m, M$ and $\nob$.
\end{theorem}

\begin{remark}
    As we will see from the proof, one may conclude (\ref{eq:thm_nearineq}) but with the prefactor $\frac{C_1}{\mathrm{gap}_+}$ replaced with a constant multiple of $\big( \int_{\Gamma^\star} \big[ \ep_{N_0+1}(\bm \xi) - \ep_{N_0}(\bm \xi)\big]^{-2} \mathrm{d}\bm\xi\big)^{1/2}$. 
\end{remark}

\begin{remark}[Existing results]
\label{rem:CT}
We briefly note here that (\ref{eq:thm_nearineq}) but with $\mathsf{gap}_+$ replaced with $\mathsf{gap}_-$ follows from a contour integral representation of the density matrix, together with Combes--Thomas type estimates on the resolvent \cite{benzi2013decay,  chen2016qm, weinan2010electronic, goedecker1998decay, ortner2020locality}. 
That is, we let $\mathscr C$ be a simple, closed, positively oriented contour encircling $\sigma(\Ham^{\mathrm{ref}}) \cap (-\infty,\ep_{\mathrm{F}})$ (depicted in Figure~\ref{fig:contour}), write
\begin{align}
    \label{eq:dm_real}
    \dmr
    = \oint_{\cc} \chi_{(-\infty, \eF)}(z)\big(z - \Hr\big)^{-1}\frac{\mathrm{d}z}{2\pi i}
    = \oint_{\cc} \big(z - \Hr\big)^{-1}\frac{\mathrm{d}z}{2\pi i},
\end{align}
and use the following Combes--Thomas estimate:  
for $\mathfrak d^{\rm ref} \coloneqq \mathrm{dist}\big( z, \sigma(\Ham^{\mathrm{ref}}) \big) > 0$, we have
\begin{align}
    \left|\big(z - \Ham^{\mathrm{ref}}\big)^{-1}_{\ell k}\right|
    \leq \frac
        {2}
        {\mathfrak{d}^{\rm ref}}
    e^{-\gamma_{\mathrm{CT}}(\mathfrak{d}^{\rm ref}) r_{\ell k}}
    \label{eq:CombesThomas}
\end{align}
where 
$\gamma_{\mathrm{CT}}(\mathfrak{d}^{\rm ref}) = c_0 \frac{\gamma_0}{h_0} \min\{ h_0, \gamma_0^d \, \mathfrak{d}^{\rm ref}\}$
with 
$c_0 = c_0(d,\mathfrak{m})$ 
depending only on the geometry. 
\end{remark}

\begin{figure}
    \centering
    \includegraphics[width = 13 cm]{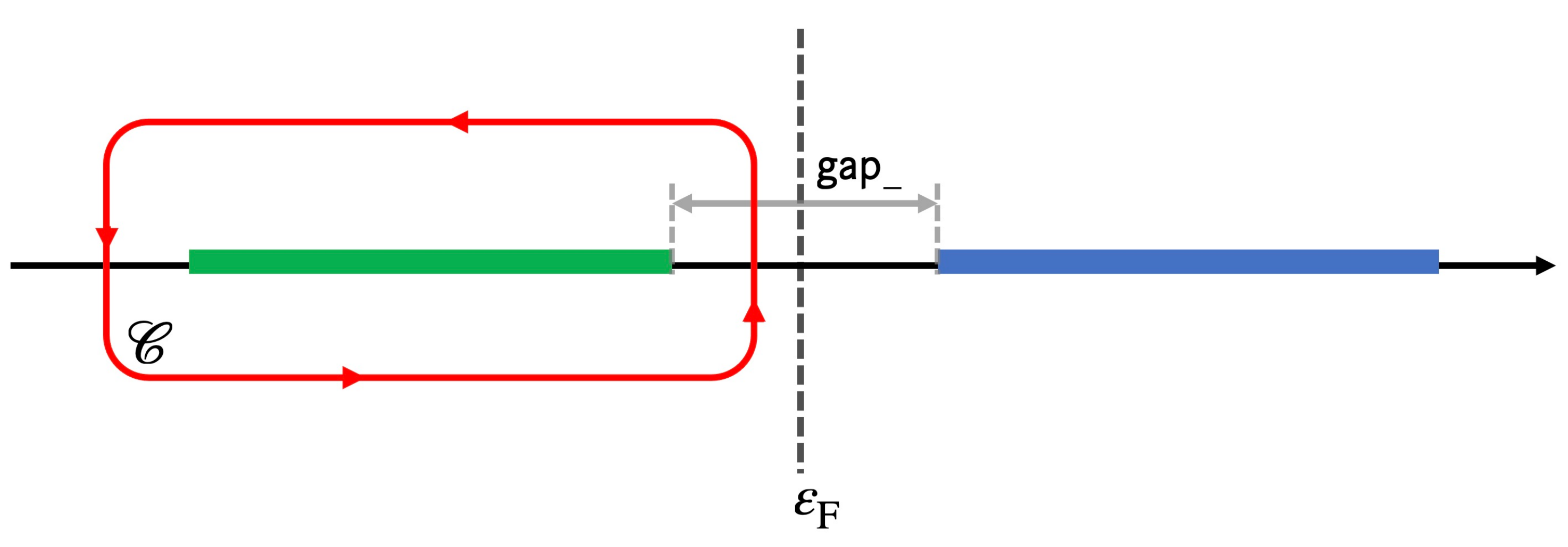}
    \caption{A schematic plot of the contour $\cc$ encircling the valence bands denoted by the green strip, and the conduction bands denoted in blue.}
    \label{fig:contour}
\end{figure}

\begin{proof}[Sketch of the proof of Theorem~\ref{thm:nearsight}] 
As in (\ref{eq:dm_real}), we consider a contour $\mathscr C$ separating the valence bands from the conduction bands, but now use the Bloch decomposition of the resolvent (see \Cref{appd:bloch_trans}), to write
\begin{align}
    \label{eq:dm_k}
       \dmr_{\ell k} &= 
    \oint_{\mathscr C}
    \bigg[
    \dashint_{\Gamma^\star}
        \left( z - \Hxi \right)^{-1} 
        e^{i \bm \xi \cdot (\ell - k)}
    \mathrm{d}\bm{\xi}
    \bigg]_{\ell_0 k_0}
    \frac
        {\mathrm{d}z}
        {2\pi i} \nonumber\\
    &= \dashint_{\Gamma^*}
    \bigg[
        \oint_{\mathscr{C} + \nu(\bxi)} 
       \left(z - \Hxi \right)^{-1}
       \frac
        {\mathrm{d}z}
        {2\pi i} 
    \bigg]_{\ell_0 k_0} e^{i\bxi\cdot (\ell - k)} \dd\bxi 
       \eqqcolon \dashint_{\Gamma^*}\big[\widehat{\rho}(\bxi)\big]_{\ell_0k_0}\, e^{i \bxi\cdot (\ell - k) }\mathrm{d}\bxi,
\end{align}
where $\ell_0, k_0 \in \Lambda_0$ satisfy $\ell - \ell_0, k - k_0 \in \A\Z^d$ 
and $\nu: \Gamma^* \rightarrow \C$ is such that $\mathscr C + \nu(\bm \xi)$ encircles $\{\ep_{n}(\bm\xi)\}_{n\leq N_{0}}$ and avoids $\{\ep_{n}(\bm\xi)\}_{n\geq N_{0}+1}$. 

We conclude \eqref{eq:thm_nearineq} by showing that $\nu$ and thus $\widehat{\rho}$ extend to analytic functions on a strip of width proportional to $\Gap$, and then apply a suitable Paley–Wiener estimate.
Full details are provided in \Cref{proof:thm_nearsight}.
\end{proof}

We now establish the stability of the nearsightedness estimates under lattice perturbations. Specifically, we extend our results to inhomogeneous systems characterized by finite-energy displacements. Since those perturbations must have finite energy-norm, they are necessarily localized, but with relatively slow (e.g., algebraic) decay~\cite{chen2019geometry}.

\begin{theorem}[Local perturbations]
\label{thm:prtb}
    Suppose that $u \in \adm(\Lambda)$ and that $\Ham(u) $ satisfies assumption $\asTB$ ~with $\eF \notin \sigma(\Ham(u))$. Then, there exist constants $C_1, C_{\ell k}(u), \etaP, \etaM > 0$ such that
\begin{align}
    \label{eq:FH_l2_pertb}
        \big|\dmu_{\ell k}\big|  
        \leq
        {\frac{C_1}{\mathsf{gap}_+}} 
        e^{-\etaP r_{\ell k}}
        + C_{\ell k}(u) e^{-\etaM r_{\ell k}},
        \quad 
        \forall \ell, k \in \Lambda,
\end{align}
where $\etaP := c_1 \frac{\gamma_0}{h_0} \min\left\{h_0,\gamma_0^{d} \, \mathsf{gap}_+\right\}$ and $\etaM \coloneqq c_2 \frac{\gamma_0}{h_0}\min\{h_0, \gamma_0^d\gap\}$. Here, the constants $c_1, c_2, C_1 > 0$ depend only on $d, \m, M$ and $\nob$, and $C_{\ell k}(u)$ depends on the norm $\|Du\|_{\ell^2_{\Upsilon}}$ and satisfies $C_{\ell k}(u) \to 0$ as $|\ell| + |k| - r_{\ell k} \to \infty$. 
\end{theorem} 

\begin{remark}
The perturbation induces a site-dependent constant $C_{\ell k}(u)$, which satisfies
\begin{align*}
	C_{\ell k}(u)
	    \lesssim
	        \|Du\|_{\ell^2_{\Upsilon}(\Lambda\backslash B_{R})}
	        +
	        \|Du\|_{\ell^2_{\Upsilon}(B_{R})}
	        e^{\etaM R}e^{-\etaM(|\ell| + |k| - r_{\ell k})}
\end{align*}
for some sufficiently large $R >0$. Consequently, 
in the asymptotic regime where $|\ell| + |k| - r_{\ell k} \to \infty$, the constant $C_{\ell k}(u)$ becomes asymptotically independent of the lattice sites. Moreover, in the limit of vanishing perturbation (i.e., $\|Du\|_{\ell^2_{\Upsilon}(\Lambda)} \to 0$), the above estimate converges to the homogeneous bound stated in \Cref{thm:nearsight}. 
\end{remark}

\begin{proof}[Sketch of the proof] 
We split $\left|\dmu_{\ell k}\right|$ into two parts by the triangle inequality
\begin{align*}
    \left|\dmu_{\ell k}\right| \leq \left|\dmr_{\ell k}\right| + \left|\dmu_{\ell k} - \dmr_{\ell k}\right|,
\end{align*}
where we have 
\begin{align}
\label{eq:dm_diff_u}
    &\left|\dmu_{\ell k} - \dmr_{\ell k}\right| = \left| \oint_{\cc} \left[\big(z - \Ham(u)\big)^{-1}\big(\Ham(u) - \Hr\big)\big(z - \Hr\big)^{-1}\right]_{\ell k}\frac{\dd z}{2\pi i } \right|.
\end{align}
Then apply the Combes--Thomas estimate \eqref{eq:CombesThomas} and its generalization in \eqref{eq:dm_diff_u} to conclude.
For full details, see \Cref{proof:prtb}.
\end{proof}

\begin{remark}
\label{rmk: maxNorm_prtb}
Analogously, we can also consider small global displacements $u$ in the sense that the max-norm,
\[
\|D u\|_{\ell_{\infty}}:=\sup _{\ell \in \Lambda} \sup _{\rho \in \Lambda-\ell} \frac{\left|D_\rho u(\ell)\right|}{|\rho|},
\] 
is sufficiently small ~\cite{chen2016qm, ortner2020locality}. The additional term in \eqref{eq:FH_l2_pertb} can be replaced by a site-independent estimate $C(u) e^{-\etaM r_{\ell k}}$ for some constant $C(u)>0$ depending on $\|Du\|_{\ell_{\infty}}$.
\end{remark}

\medskip
As noted previously, there are two types of locality: weak and strong. While the nearsightedness reflecting weak locality of the density matrix is well understood, it is generally insufficient for characterizing the locality of interatomic interactions. This is important in the construction of interatomic potential models or multi-scale schemes such as QM/MM coupling methods \cite{chen2019adaptive, chen2016qm, chen2017qm, chen2022qm, csanyi2005multiscale}, where selecting an appropriate interaction cut-off radius is essential, for example. To justify those schemes, the strong locality  condition is required. 

Following ~\cite{chen2016qm,ortner2020locality}, one can derive a $\gap$ dependent estimate for the derivative of the density matrix, that is, for $u \in \adm(\Lambda)$,  $1 \le j \le d$, we have
\begin{align}
    \label{eq:strong}
    \left|
        \frac{\partial[\dmu]_{\ell \ell}}{\partial [u(k)]_{j}}     
        	\bigg|_{u = \mathbf{0}}
    \right| 
        \lesssim 
            e^{-\etaM r_{\ell k}},
            \quad
            \forall 
            \ell , k \in \Lambda, 
\end{align}          
where $\etaM$ is introduced in \eqref{eq:FH_l2_pertb}.

A natural question is whether the strong locality can also be improved with respect to $\Gap$.
Unfortunately, the answer is negative. Here, we provide a heuristic argument explaining why no stronger result can be expected. 
First, we note that
\begin{align}
\frac{\partial\rho_{\ell\ell}}{\partial [\bm r_k]_{j}} 
&\coloneqq
\frac{\partial\big[\dmu\big]_{\ell\ell}}{\partial [u(k)]_{j}}\Bigg|_{u=\mathbf{0}} 
    =
    \oint_{\cc} \left[\big(z - \Ham^{\rm ref}\big)^{-1} \, \frac{\partial \Ham(u)}{\partial [u(k)]_j}\bigg|_{u = \mathbf{0}}\, \big(z - \Ham^{\rm ref}\big)^{-1} \right]_{\ell\ell}
    \frac{\dd z}{2\pi i } \nonumber\\
    &= \sum_{mn}
    \frac
        {\partial \Ham_{mn}}
        {\partial [\bm r_k]_j}
        :
    \fint_{(\Gamma^*)^2} 
    \left[ \oint_{\cc}
    \big(z - \Ham_{\bxi}\big)^{-1}_{\ell_0 m_0} 
    \big(z - \Ham_{\bzeta}\big)^{-1}_{n_0\ell_0}
    \frac{\dd z}{2\pi i }
    \right]
    e^{i \bxi \cdot (\ell - m)}
    e^{i \bzeta \cdot (n - \ell) }
    \dd\bxi \dd\bzeta. 
    \label{eq:drho-drk}
\end{align}
Therefore, as in the proof of Theorem~\ref{thm:nearsight}, one may obtain the strong locality estimate from the locality of the Hamiltonian, as in \asTB, and the analyticity of the functions
\begin{align}
 \label{eq:rhohat2}
	(\bxi,\bzeta) \mapsto \oint_{\cc}
    \big(z - \Ham_{\bxi}\big)^{-1}_{\ell_0 m_0} 
    \big(z - \Ham_{\bzeta}\big)^{-1}_{n_0\ell_0}
    \frac{\dd z}{2\pi i }.
\end{align}
Recall that, for the analogous function $\widehat{\rho}(\bm \xi)$ \eqref{eq:dm_k} used in the nearsightedness estimate, one may extend the region of analyticity of $\widehat{\rho}$ to a strip of width proportional to $\Gap$ by replacing $\mathscr C$ with a contour depending on $\bxi$. 

However, this is not the case for the functions \eqref{eq:rhohat2}. Since the contour $\mathscr C = \mathscr C(\bxi,\bzeta)$ must simultaneously separate $\bigcup_{n\leq N_0}\big\{ \ep_{n}(\bxi)\big\}$ from $\bigcup_{n\geq N_0+1}\big\{ \ep_{n}(\bxi) \big\}$ and $\bigcup_{n\leq N_0}\big\{ \ep_{n}(\bzeta)\big\}$ from $\bigcup_{n\geq N_0+1}\big\{ \ep_{n}(\bzeta) \big\}$, the region of analyticity of \eqref{eq:rhohat2} becomes a subset of 
\begin{align}
	\label{eq: deriv_region}
    \left\{ 
    ({\bm \eta}_1, {\bm \eta}_2) \in (\mathbb R^d)^2 \colon  
    \ep_{N_0}(\bxi + i{\bm \eta}_1)
    \not= \ep_{N_0 + 1}(\bzeta + i{\bm \eta}_2), 
    ~\forall (\bxi, \bzeta) \in (\Gamma^*)^2
    \right\},
\end{align}
and this region depends on $\gap$. Additional empirical evidence from our numerical experiments (see \Cref{fig:strong_loc} and \Cref{fig:Mg2Si_strong} in Section \ref{sec:numerical}),
also strongly supports that the strong locality does indeed depend on $\gap$ rather than $\Gap$.

Nevertheless, we show that the dependence of the strong locality estimate on the spectral gap
can be further improved in the small-gap regime,
provided that the energy bands near the Fermi level satisfy additional regularity conditions.
More precisely, under suitable analyticity assumptions, the decay rate improves from a linear
dependence on $\gap$ to a square-root scaling $\sqrt{\gap}$.

Specifically, we impose the following regularity assumption on the relevant energy bands:
\smallskip
\begin{flushleft}
	\asReg ~Suppose that
	the bands $\ep_{N_0}$ and $\ep_{N_0 + 1}$ admit an analytic continuation to the domain
	\[
	\C_{\reg}^d \coloneqq \{\bxi + i\bzeta \in \C^d: |\bzeta|< \reg\} 
	\]
	for some
	constant $ \reg > 0$.
\end{flushleft}
\smallskip
Under this assumption, we obtain the following improved bound:

\begin{theorem}
\label{thm: sharp_estimate}
	Suppose that $u \in \adm(\Lambda)$ and that $\Ham^{\rm ref}, \Ham(u)$ satisfies assumptions $\asTB$ and $\asReg$, with $\eF \notin \sigma(\Ham(u))$. Define 
	\begin{align*}
		m_{\infty} \coloneqq \max_{n \in \{N_0, N_0 + 1\} }\sup_{{\bm \omega} \in 
			\C_{\reg/2}^d \cap \C_{\gamma_0/2}^d}
		  	\|\nabla^2\ep_{n}({\bm \omega})\|_{\rm F}.
	\end{align*}
	Then,
	\begin{itemize}
		\item[(i)](Nearsightedness) There exist constants $C_2,\widetilde{\eta}_+>0$ such that
		\begin{align}
	        \label{eq:shaper_nearsight}
	            \left|\dmr_{\ell k}\right| 
	            \leq 
	            \frac{C_2}{\Gap} e^{-\widetilde{\eta}_+ r_{\ell k}},
	            \quad
	            \forall \ell, k \in \Lambda;
	    \end{align}
		\item[(ii)] (Strong locality) Moreover, there exist constants $C_3,\widetilde{\eta}_->0$ such that, for any $1\le j\le d$, 
		\begin{align}
	        \label{eq:shaper_strong_loc}
	            \left|
	                \frac{\partial[\dmu]_{\ell \ell}}{\partial [u(k)]_{j}}
	                 {\bigg|_{u = {\bf 0}}}
	            \right| 
	            \leq 
	            \frac{C_3}{\gap^2} e^{-\widetilde{\eta}_- r_{\ell k}},
	            \quad
	            \forall \ell, k \in \Lambda,
	    \end{align}
	\end{itemize}
where $\widetilde{\eta}_+ \coloneqq c_3 \min\left\{\gamma_0, \reg, \,\sqrt{\Gap/m_{\infty}}\right\}$ and $\widetilde{\eta}_- \coloneqq c_4 \min\left\{\gamma_0, \reg, \,\sqrt{\gap/m_{\infty}}\right\}$ for constants $c_3, c_4 > 0$.  Here, $C_2 > 0$ depends only on $d, \m, M$ and $\nob$, and $C_3 > 0$ depends on $h_0, \gamma_0, d, \m, M$ and $\nob$.
\end{theorem}

\begin{proof}[Sketch of the proof] 
To characterize a suitable analytic subregion of \eqref{eq: deriv_region}, it suffices to analyze the behavior of the real parts of the relevant energy bands. By the analyticity assumption $\asReg$, the band functions $\ep_n$ admit a uniform bound on their Hessians over $\C_{\reg/2}^d$, which ensures that $m_{\infty} < \infty$.
Consequently, for any $\bxi + i\bzeta \in \C_{\reg/2}^d$ and $n \in \{N_0, N_0 + 1\}$, we obtain 
\begin{align}
	\label{eq:square_pertb}
    \re\ep_{n}(\bxi + i \bzeta)
	\in 
	B_{m_{\infty}|\bzeta|^2}\!\left(\ep_{n}(\bxi)\right),
\end{align}
where the linear term $i \nabla \ep_n(\bxi)\cdot \bzeta$ is purely imaginary and therefore does not contribute to the real parts.

Combining this quadratic control of the band separation with the Paley--Wiener estimate, we conclude that the width of the analytic region for \eqref{eq:rhohat2} scales like $\sqrt{\gap/m_{\infty}}$ in the small-gap regime, and the desired exponential decay estimate follows. The same argument applies to $\rho^{\rm ref}_{\ell k}$
 yielding an analogous exponential decay rate scaling with $\sqrt{\Gap/m_{\infty}}$. 
For full details, see \Cref{proof: sqrt_pf}.
\end{proof}

\begin{remark}
\label{rem:sqrt-nearsightedness}
    The assumption $\asReg$ is generally quite strong but they do apply to certain cases, such as those shown in \Cref{fig:1d_bands_type11}(a), (b) and \Cref{fig:1d_bands_type2}. 
    This condition may fail when band crossings occur in $d\geq 2$~\cite{Nenciu1991, panati2007triviality}. 
    A typical case are \textit{Dirac} cones: 
    \[
    \Ham_{\bxi} = 
    \begin{pmatrix}
        0 & \xi_1 - i \xi_2 \\
        \xi_1 + i \xi_2 & 0
    \end{pmatrix},
    \quad \bxi = (\xi_1, \xi_2) \in \Gamma^*,
    \]
    where the energy bands $\ep_\pm(\bxi) = \pm \sqrt{\xi_1^2 + \xi_2^2}$ are not even differentiable at $\bxi = {\bf 0}$.
    The analysis of more complex systems will be left for our future work.
\end{remark}

\begin{remark}[General observables] 
\label{rmk: general_ob}
More generally, we suppose $O: \R \to \R$ is a function that extends analytically to an open neighbourhood of the spectrum. Then, one can define the corresponding local quantities for all $\ell, k \in \Lambda$ by
\begin{align*}
    O(u)_{\ell k} :=  \oint_{\cc}O(z) \big(z - \Ham(u)\big)_{\ell k}^{-1}\frac{\mathrm{d}z}{2\pi i},
\end{align*}
and the observable corresponding to $O$ is given by $\sum_{\ell}O(u)_{\ell\ell} := \tr\left(O\big(\Ham(u)\big)\right)$. The derivatives of $O_{\ell\ell}$ related to mechanical response are given by
\begin{align*}
    \frac{\partial O(u)_{\ell\ell} }{\partial [u(k)]_j}\bigg|_{u = {\bf 0}} = 
    \oint_{\cc} O(z) \left[(z - \Ham^{\rm ref})^{-1} \frac{\partial \Ham(u)}{\partial [u(k)]_j}\bigg|_{u = {\bf 0}} (z - \Ham^{\rm ref})^{-1}  \right]_{\ell\ell} \frac{\dd z}{2\pi i}.
\end{align*}
All of our conclusions can be extended to this more general case, which also improve the results of \cite{chen2016qm, ortner2020locality}.
\end{remark}

\section{Numerical Experiments}
\label{sec:numerical}
In this section, we present numerical experiments to support our analytical results. Specifically, we demonstrate the dependence of nearsightedness on $\Gap$ and its stability under perturbations, as well as the strong locality dependence on $\gap$. We also include experiments that validate the potential sharp estimates. To this end, we construct 1D toy models with tunable band structures, allowing us to adjust $\Gap$ and $\gap$ in a controlled manner. In addition, we also illustrate our results on a realistic 3D system, $\rm Mg_2Si$.

\subsection{1D Toy Model}
\label{subsec:toy model}
Let us consider an infinite 1D chain $\Lambda := \Z$. We construct the tight binding Hamiltonian $\Ham(u)$ by assigning each atom with two atomic orbitals. Only nearest-neighbour interactions are included.

We suppose that $\Ham(u)$ is given by $\Ham(u)_{\ell k} = h\big( \bm r_{\ell k}(u) \big)$ with 
\begin{align}
    h\big( \bm r_{\ell k}(u) \big)
    &= \begin{cases}
        \bm c &\text{if } \ell = k \\
        \bm f\big( r_{\ell k}(u) \big)
        &\text{if } \ell \not= k
    \end{cases}
\end{align}
and $f_{ij}( r ) \coloneqq (a_{ij} + b_{ij} r) \eta(r)$ for parameters $\bm a, \bm b, \bm c \in \mathbb R^{2\times 2}_{\rm sym}$ and for some smooth cut off function $\eta$ with $\eta = 1$ on $[0,1]$ and $\eta = 0$ on $[2,\infty)$. One may verify that $\Ham(u)$ satisfies \asTB~with Bloch transform of $\Ham(0)$ given by
\begin{align}
    \Ham_\xi = \bm c + 2 \bm f \cos(\xi), 
    \quad \text{with } \xi \in \Gamma^\star = [-\pi, \pi].
\end{align}
where $\bm f \coloneqq \bm f(1)$. By solving the eigenvalue problem $\Ham_{\xi} u_n(\xi) = \ep_{n}(\xi) u_n(\xi)$,
we obtain two energy bands $\ep_\pm(\xi)$, as well as an analytic expression for $\Gap$,
which enables us to build a chain with varying $\gap$ while keeping $\Gap$ approximately fixed; cf.~the vertical gap in \Cref{fig:gap_type}. 

In the following, we will consider two ``types'' of band gaps: \textit{(i)} first, we consider $c_{12} = c_{21} = 0$ and take the off-site block to be a full matrix. We also consider \textit{(ii)} the case when the on-site block is approximately diagonal ($c_{12} \approx 0$), while the off-site block is diagonal (that is, $f_{12} = f_{21} = 0$) and thus $\ep_{\pm}(\xi)$ resemble cosine-like bands. The following figures illustrate band structures of the two types Hamiltonian under various $\gap$ and $\Gap$ settings.

\textit{(i)} In \Cref{fig:1d_bands_type11} (a) and (b), we fix $\Gap = 1.0$ and let $\gap$ decease from $1.0$ to $0.1$, showing an indirect gap, with the curvature norm $m_{\infty}$ uniformly bounded. 
This means we would expect $\sqrt{\gap}$ dependence in the derivatives of $\rho(u)$, as predicted by Theorem~\ref{thm: sharp_estimate} (c.f.~\Cref{fig:strong_loc}). 
Notably, as $\Gap \to 0$, $m_{\infty}$ diverges, see~\Cref{fig:1d_bands_type11} (c).

\begin{figure}[H]
    \centering
    \subfigure[$\Gap = \gap = 1.0$]{
    \includegraphics[width = 5.0cm]{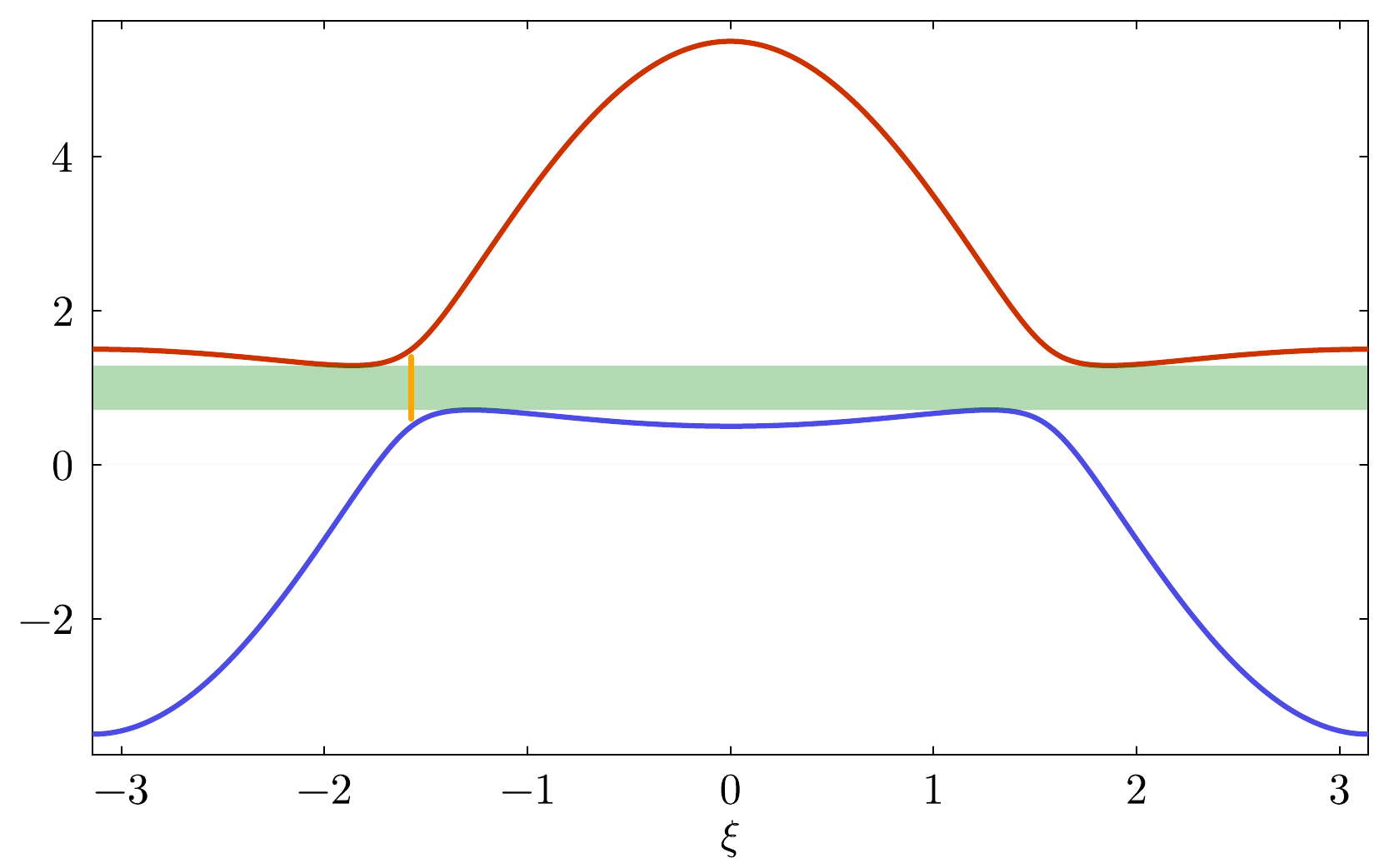}
    }
     \subfigure[$\Gap = 1.0, ~\gap = 0.1$]{
    \includegraphics[width = 5.0cm]{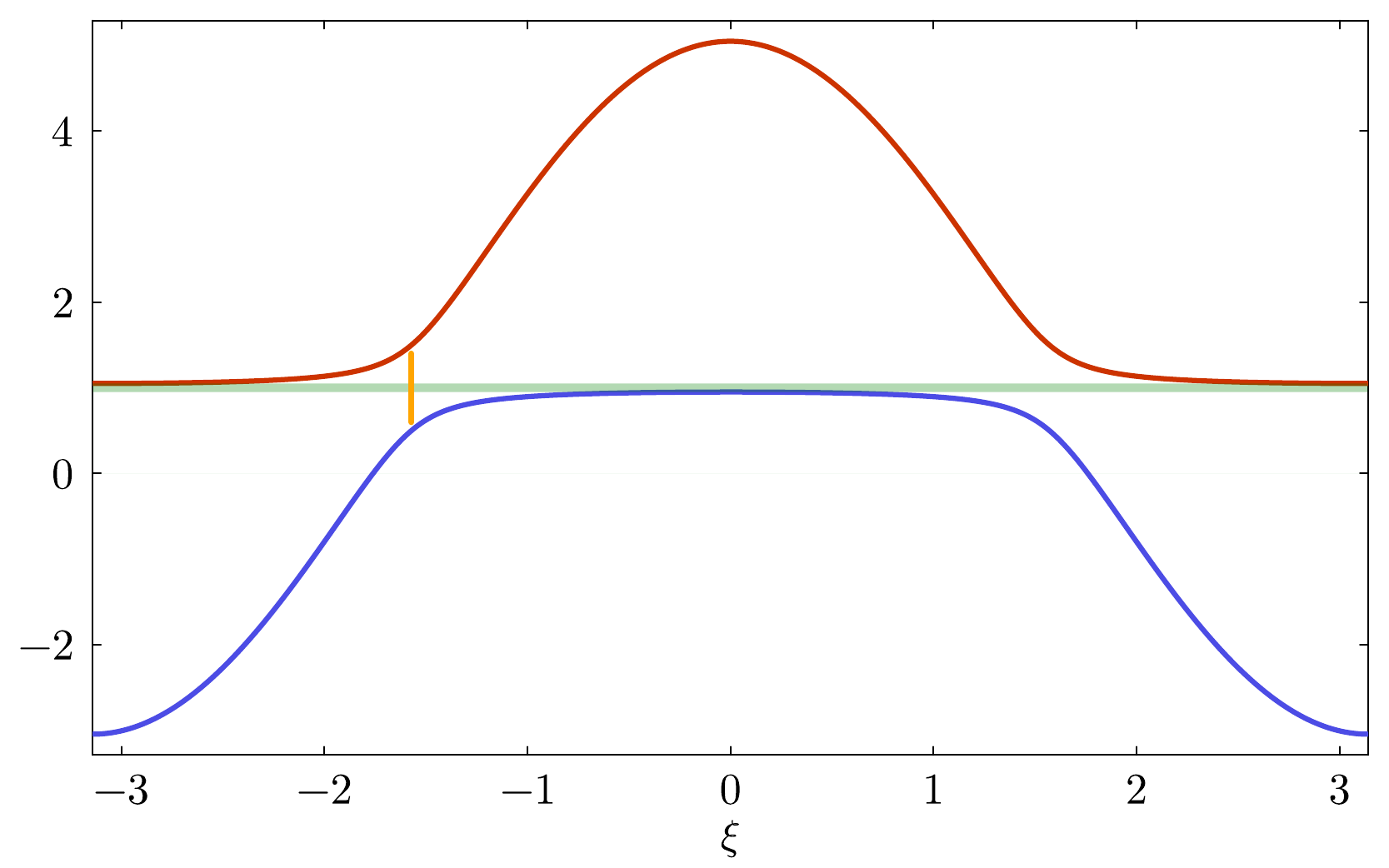}
    }     
     \subfigure[$\Gap =0.25,  ~\gap = 0.01$]{
    \includegraphics[width = 5.0cm]{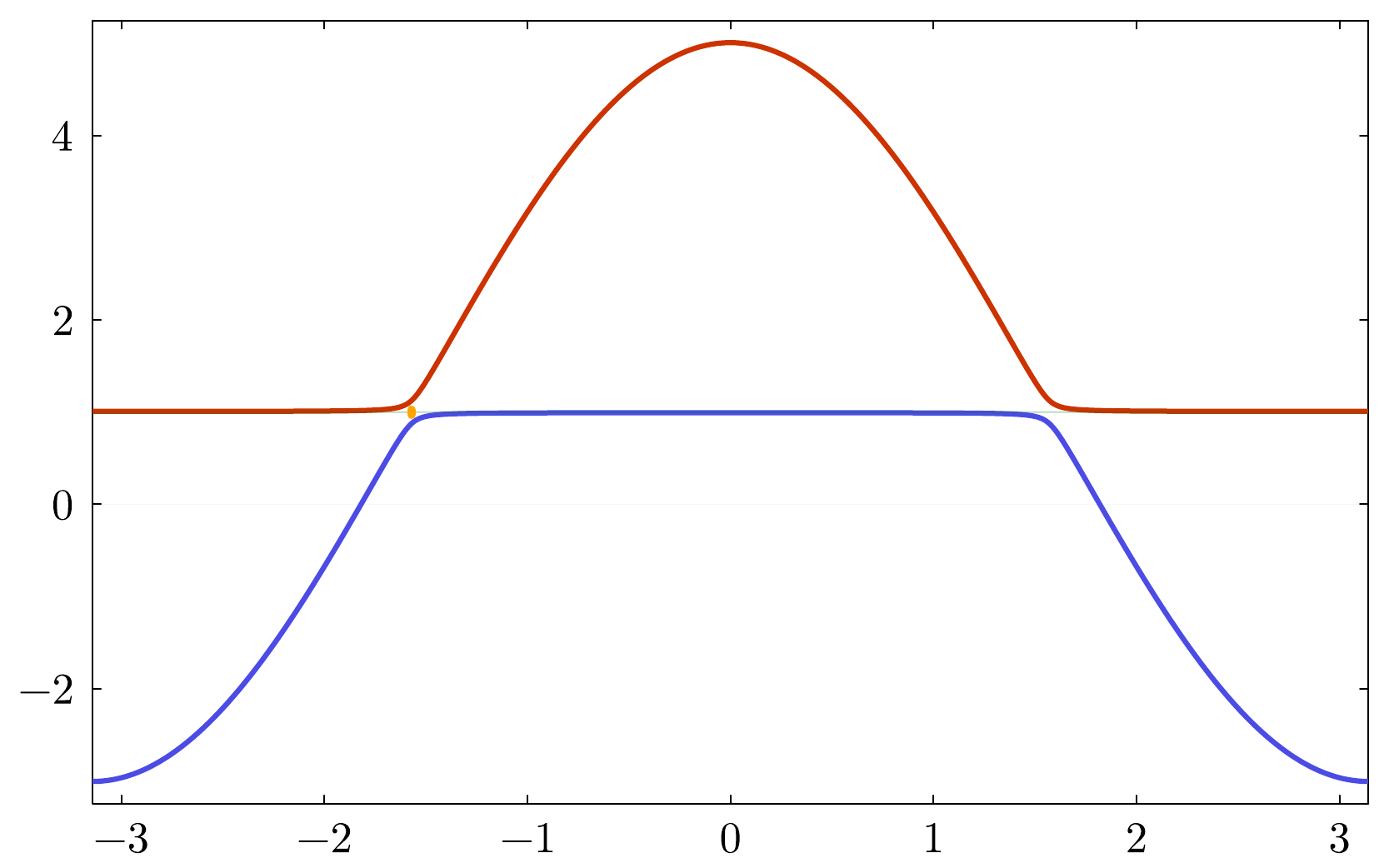}
    }
    \caption{Band structure of type \textit{(i)} model.}
    \label{fig:1d_bands_type11}
\end{figure}

\textit{(ii)} \Cref{fig:1d_bands_type2} displays the cosine-like bands with a direct band gap. As $\Gap \to 0$, $m_{\infty}$ remains uniformly bounded and one expects a $\sqrt{\Gap}$ dependence in the nearsightedness estimate as predicted by  Theorem~\ref{thm: sharp_estimate} (cf.~\Cref{fig: sqrt_Gap}).

\begin{figure}[H]
    \centering
    \subfigure[$\Gap = \gap = 1.0$]{
    \includegraphics[width = 5.0cm]{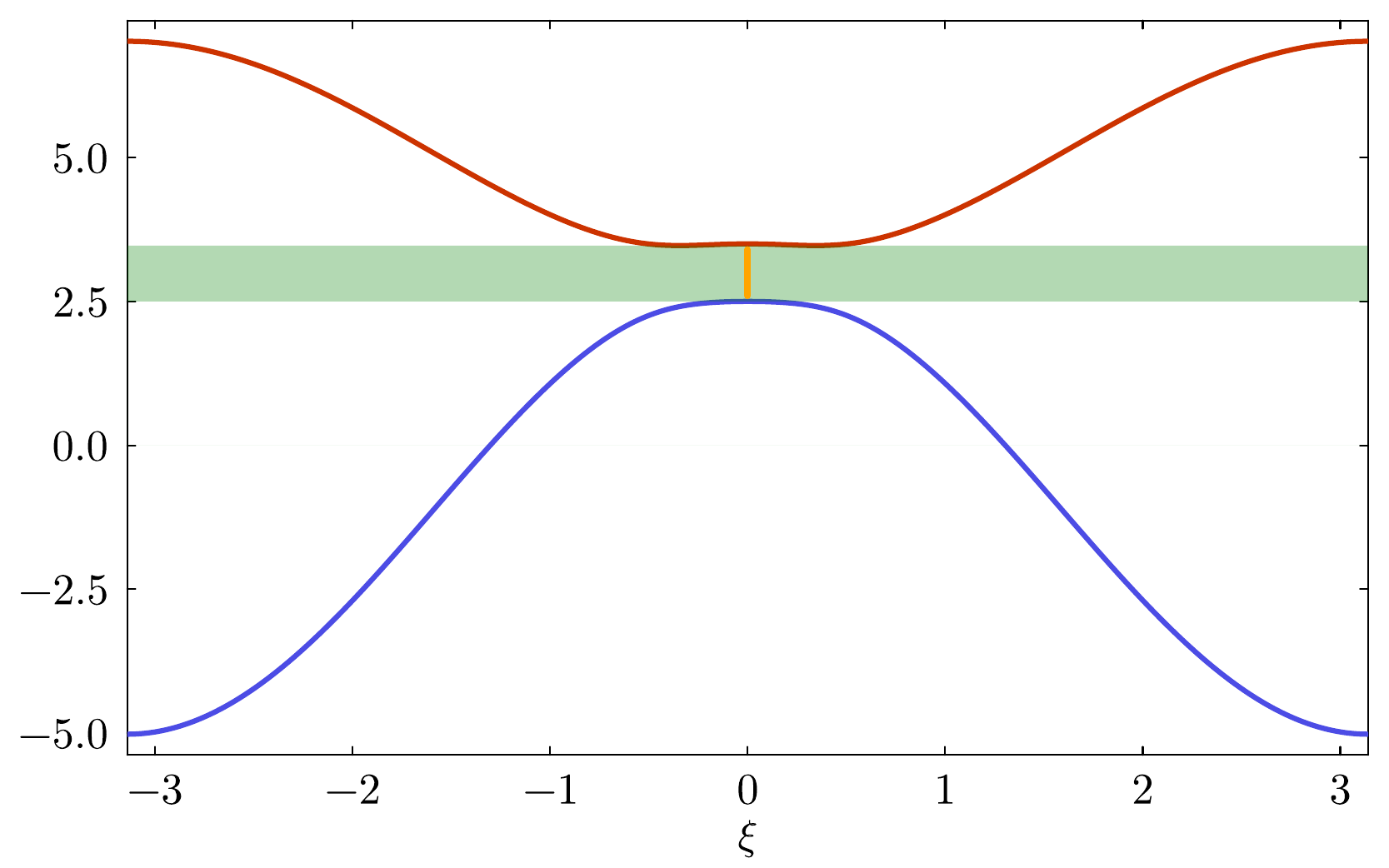}
    }
    \subfigure[$\Gap = \gap = 0$]{
    \includegraphics[width = 5.0cm]{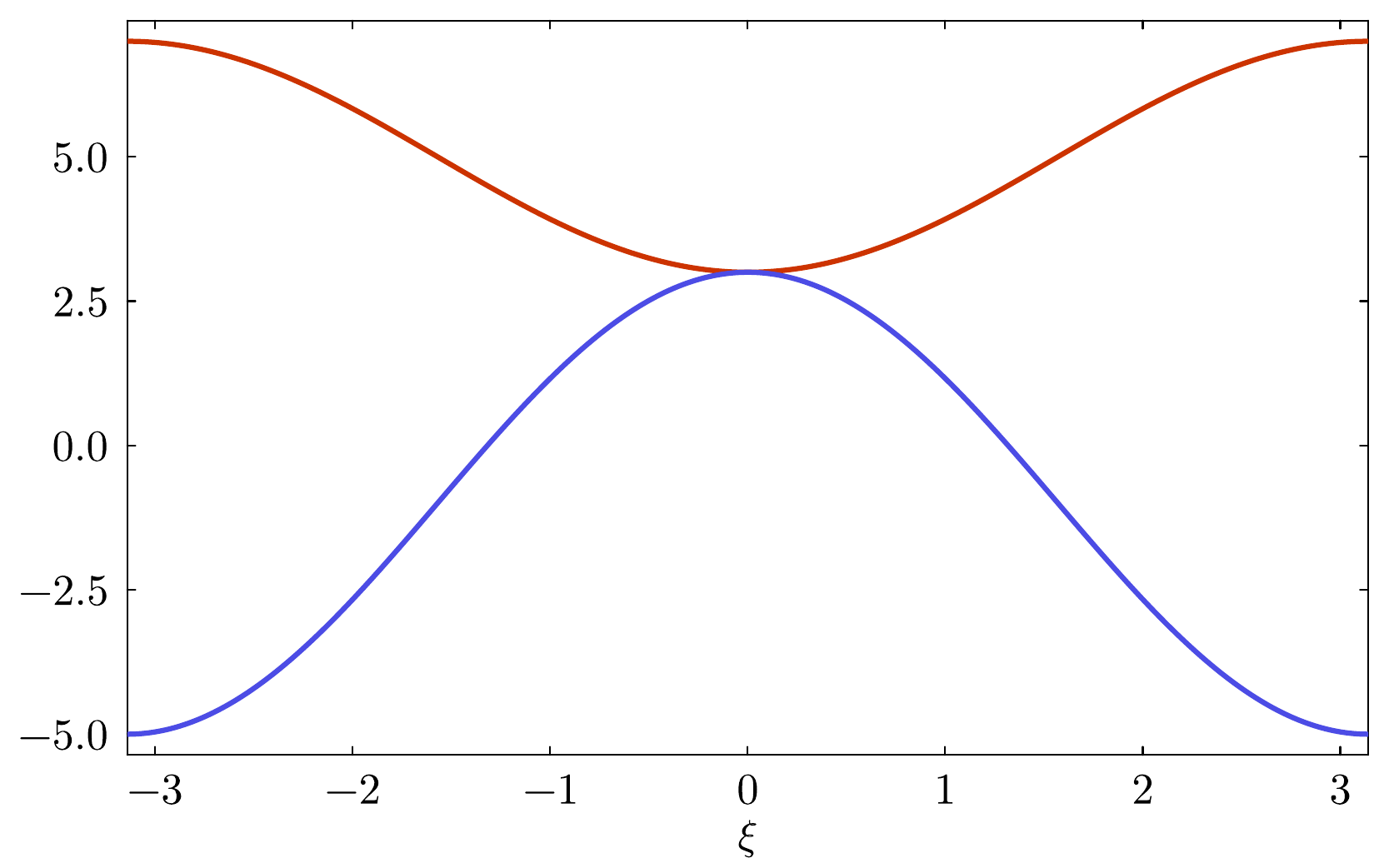}
    }
    \caption{Band structure of type \textit{(ii)} model.}
    \label{fig:1d_bands_type2}
\end{figure}

The specific examples and computational details can be found in \Cref{appd: 1d_chain}.

\subsection{Nearsightedness}
\label{eq: test_nearsihtedness}
We conduct the tests using the type \textit{(i)} model under two settings: one with fixed $\Gap$ and varing $\gap$, and the other with fixed $\gap$ and varing $\Gap$, to validate the dependence on $\Gap$.
In practice, we simulate supercell models including 200 atoms. 

In \Cref{fig:weak_loc}, we compare the decay of the density matrix by fixing $\Gap = 2.0$ and varying $\gap$ across the values $1/2, 1/8, 1/32, 1/128, 1/512$, where the decay curves align almost perfectly. Moreover, we also compare the results by fixing $\gap = 0.01$ and taking $\Gap = 2.0, 1.0, 0.5, 0.25$, which shows a linear dependence on $\Gap$ as displayed in \Cref{tab: weak_loc_slope}.

\begin{figure}[H]
    \centering
    \subfigure[]{
    \includegraphics[width = 7.2cm]{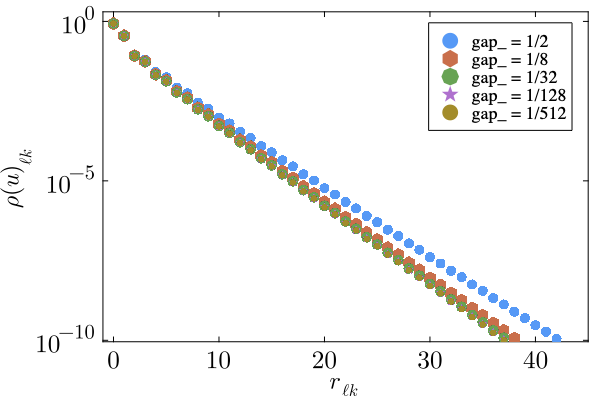}}
    \subfigure[]{
    \includegraphics[width= 7.2cm]{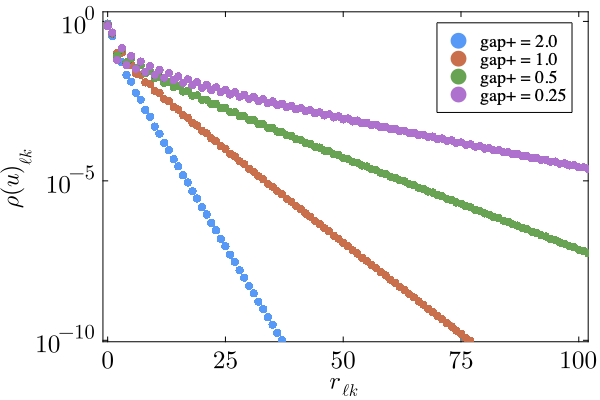}
    }
    \caption{1D chain: nearsightedness in homogeneous system with (a) fixed $\Gap = 2.0$ and (b) $\gap = 0.01$; cf.~\Cref{thm:nearsight}. 
    }
    \label{fig:weak_loc}
\end{figure}

\begin{table}[H]
    \centering
    \begin{tabular}{c|c|c|c|c}
    \hline
    $\Gap$ & 2.0 & 1.0 & 0.5 & 0.25\\ \hline
    slope  & -0.2456 & -0.1194 & -0.0649 & -0.0291 \\ \hline
    \end{tabular}
    \caption{Slopes for the decay curves in Figure~\ref{fig:weak_loc} (b): the absolute values of slopes decay by half as $\Gap$ decreases by half.}
    \label{tab: weak_loc_slope}
\end{table}


In \Cref{fig:weak_prtb_l2} and \Cref{fig:weak_prtb_inf}, we conduct tests with admissible perturbations. Specifically, subfigures (a), (b) and (c) correspond to $\gap = 0.1, 0.01$ and $0.001$, respectively. The perturbations have infinity norm of $\ep = 10^{-6}, 10^{-4}, 10^{-2}$, and $\ell^2_{\Upsilon}$ norm of $\ep = 10^{-5}, 10^{-3}, 10^{-1}$, respectively. 
The turning points of the curves in each figure shift forward gradually as $\gap$ decreases and $\ep$ increases. Additionally, we observe that the influence of $\gap$ becomes increasingly dominant as $\ep$ rises, which aligns perfectly with the analysis in \Cref{thm:prtb}.

\begin{figure}[H]
    \centering
    \subfigure[$\gap = 0.1$]{\includegraphics[width = 5.5cm]{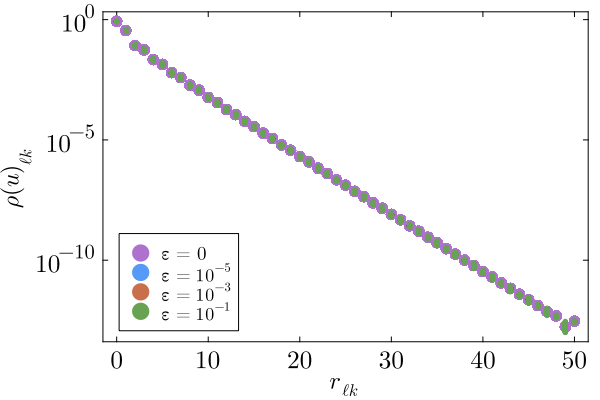}}
    \subfigure[$\gap = 0.01$]{\includegraphics[width = 5.5cm]{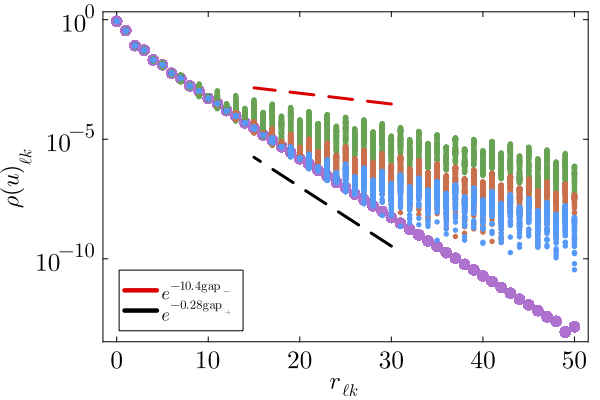}}
    \subfigure[$\gap = 0.001$]{\includegraphics[width = 5.5cm]{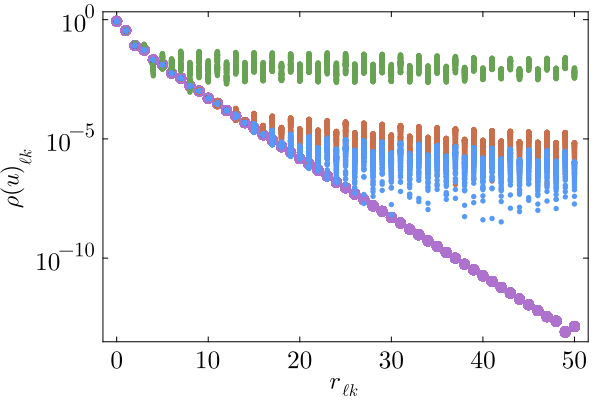}}
    \caption{1D chain: nearsightedness with local perturbations; cf. \Cref{thm:prtb}}
    \label{fig:weak_prtb_l2}
\end{figure}
\begin{figure}[H]
    \centering
    \subfigure[$\gap = 0.1$]{\includegraphics[width = 5.5cm]{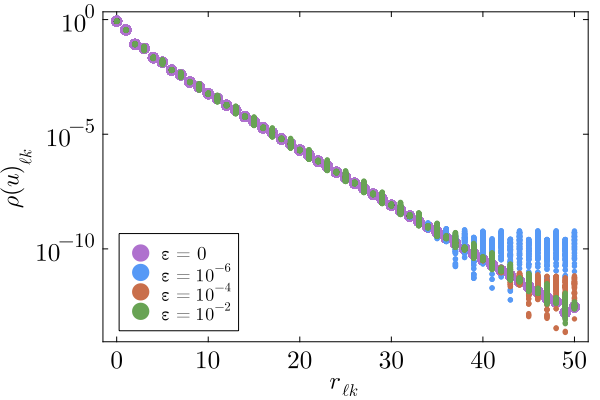}}
    \subfigure[$\gap = 0.01$]{\includegraphics[width = 5.5cm]{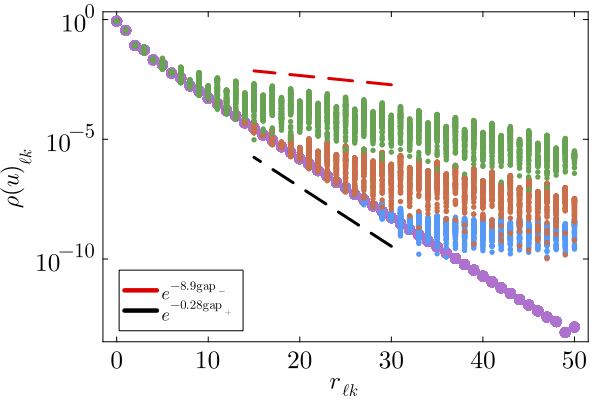}}
    \subfigure[$\gap = 0.001$]{\includegraphics[width = 5.5cm]{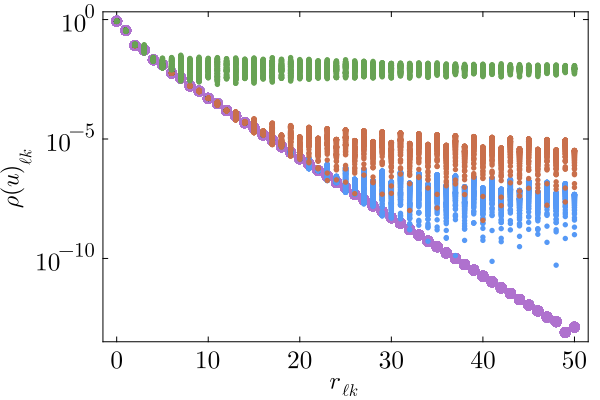}}
    \caption{1D chain: nearsightedness with global perturbations; cf. Remark~\ref{rmk: maxNorm_prtb}}
    \label{fig:weak_prtb_inf}
\end{figure}


In \Cref{fig: sqrt_Gap}, we conduct the tests using the type \textit{(ii)} model, which yields
cosine-like bands for suitable parameter choices; see \Cref{appd: 1d_chain} for details. 
We plot the decay curves of the density matrix for $\Gap(=\gap) = 0.1, 0.01, 0.001$ on a lin-log scale. The $\Gap$-slope relation in \Cref{fig: sqrt_Gap}(b) exhibits a square-root dependence on $\Gap$.

\begin{figure}[H]
    \begin{minipage}[b]{0.45\textwidth}
        \centering
        \includegraphics[width = 7.5cm]{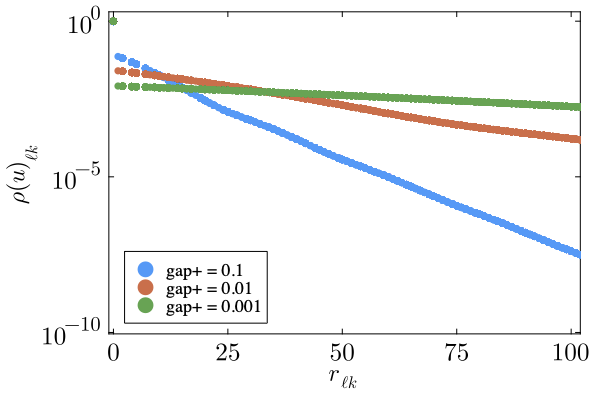}
        \caption*{(a)}
    \end{minipage}
    \begin{minipage}[b]{0.5\textwidth}
    \centering
        \begin{tabular}{c|c|c}
		\hline
		$\Gap (= \gap)$ & slope &  $c\Gap^{\alpha}$      \\ \hline
		0.1  & -0.0585 & -0.0589 \\ \hline
		0.01 & -0.0208 & -0.0205\\ \hline
		0.001 & -0.0071 & -0.0071 \\ \hline
        \end{tabular}
        \vspace{1cm}
        \caption*{(b)}
    \end{minipage}
    \caption{1D chain: the sharp estimate for nearsightedness in homogeneous system; the third column in (b) is the optimal fitting between $\Gap$ and slopes with $c = -0.17$, $\alpha = 0.47$; cf. Theorem~\ref{thm: sharp_estimate}.}
    \label{fig: sqrt_Gap}
\end{figure}

\subsection{Strong locality}
\label{eq: test_strong_locality}
In the following experiments, we focus exclusively on the type \textit{(i)} model, which allows for more controlled tuning of $\gap$ with fixed $\Gap$, enabling clearer observation of the decay behaviour.

In \Cref{fig:strong_loc}, we plot the derivative of the density matrix $\big|\frac{\partial \dmr_{\ell \ell}}{\partial r_m}\big|$ for $\ell = 1$ and $\gap = 1.0, 0.1, 0.01, 0.001$, in lin-log scales. 
We observe that the decay becomes slower as $\gap$ decreases. Additionally, we present the slopes for \Cref{fig:strong_loc}(a), which displays a square-root dependence on $\gap$.

\begin{figure}[H]
    \begin{minipage}[b]{0.45\textwidth}
        \centering
        \includegraphics[width = 7.5cm]{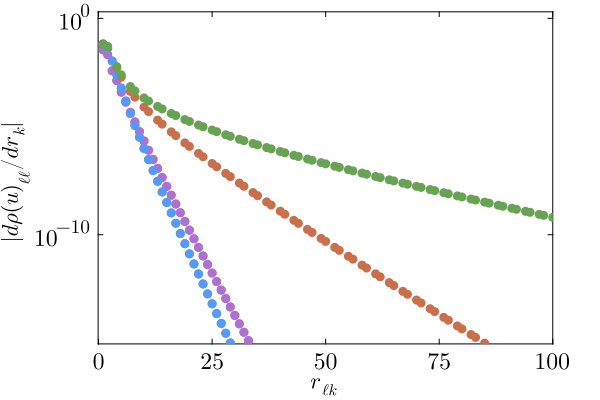}
        \caption*{(a)}
    \end{minipage}
    \begin{minipage}[b]{0.5\textwidth}
    \centering
        \begin{tabular}{c|c|c}
		    \hline
		    $\gap$ & slope &  $c\gap^{\alpha}$ \\ \hline
		    0.1   & -0.4711 & -0.4700 \\ \hline
		    0.01  & -0.1385 & -0.1451 \\ \hline
		    0.001 & -0.0556 & -0.0448 \\ \hline
		\end{tabular}
        \vspace{1cm}
        \caption*{(b)}
    \end{minipage}
    \caption{1D chain: the strong locality in homogeneous system; the third column displays the optimal fitting between $\gap$ and slope with $c = -1.51$ and $\alpha = 0.51$; cf. Theorem~\ref{thm: sharp_estimate}}
    \label{fig:strong_loc}
\end{figure}

%

To further substantiate our conclusions, we provide an example from real systems with a small $\gap$ in the next section.
\subsection{\texorpdfstring{$\bf{Mg_2Si}$}{}}
The system we investigate here is an antifluorite-type compounds, a well-known semiconductor used in thermoelectric applications, characterized by a small indirect band gap. We conduct tests using the package $\mathsf{DFTK.jl}$ \cite{herbst2021dftk}, a library of Julia routines for working with plane-wave density-functional theory algorithms. Although the discussion in the paper focuses on the tight binding model, the property of locality is not limited to this. Similar results can be found in \cite{benzi2013decay, jewski2004exact, prodan2005nearsightedness}.

The energy bands for ${\rm Mg_2Si}$ are presented in \Cref{fig:gap_bz}. The extrema of valence and conduction bands occur at the point $\Gamma$ and $\rm X$ respectively, which results in a small $\gap$ around 0.21 $eV$. However, the $\Gap$ is  relatively large around 1.8 $eV$.

We use a supercell model consisting of $4 \times 4 \times 4$ unit cells to verify the nearsightedness and strong locality. To make the results more convincing, we also present the corresponding tests for $\rm C$. In \Cref{fig:Mg2Si_weak}, (a) to (d) display the decay in homogeneous systems and the perturbed ones with infinity norm of $\ep = 0.01$. 

\begin{figure}[H]
    \centering
    \subfigure[$\rm{Mg_2Si}$: homogeneous ]{\includegraphics[width = 7.1cm]{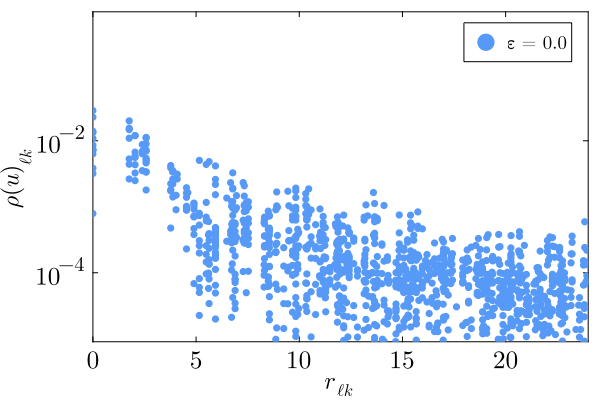}}
    \subfigure[$\rm{C}$: homogeneous]{\includegraphics[width = 7.0cm]{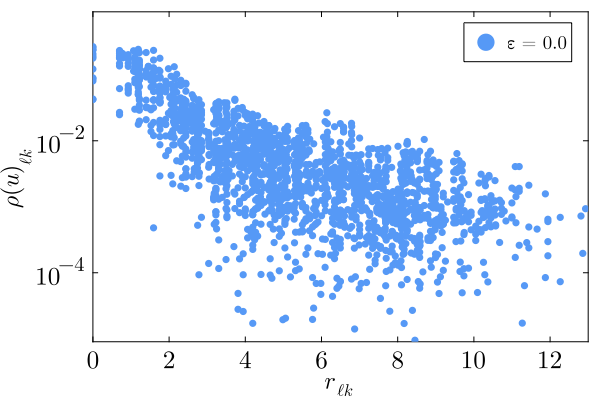}}
    \subfigure[$\rm{Mg_2Si}$: $\ep = 0.01$]{\includegraphics[width = 7.1cm]{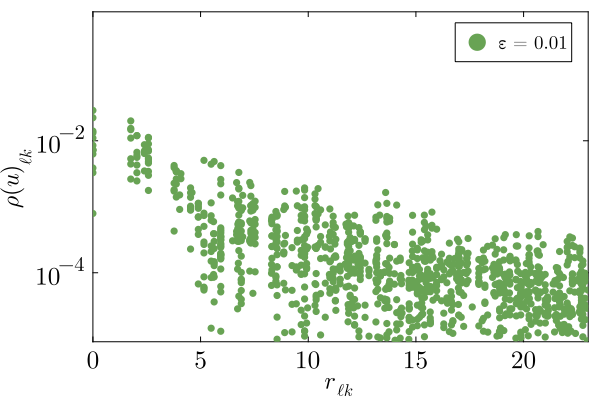}}  
    \subfigure[$\rm{C}$: $\ep = 0.01$]{\includegraphics[width = 7.0cm]{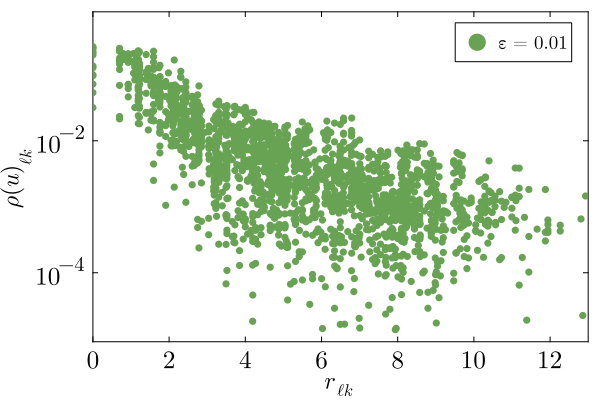}} 
    \caption{$\rm Mg_2Si$ and $\rm C$: the nearsightedness property of homogeneous and perturbed systems; cf. \Cref{thm:nearsight} and \Cref{thm:prtb}.}
    \label{fig:Mg2Si_weak}
\end{figure}

\begin{figure}[H]
    \centering
    \subfigure[$\rm{Mg_2Si}$: Energy Hessian]{\includegraphics[width = 7.0cm]{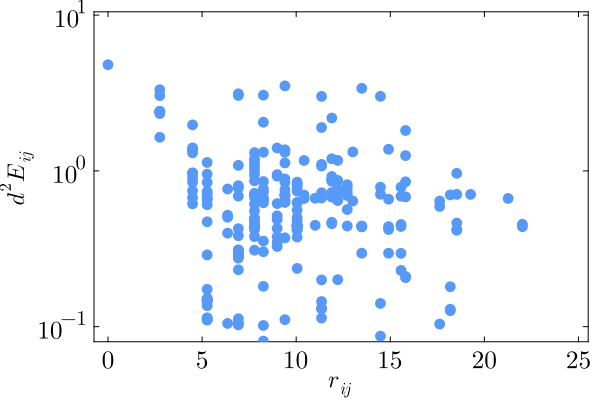}}
    \subfigure[$\rm{C}$: Energy Hessian]{\includegraphics[width = 7.0cm]{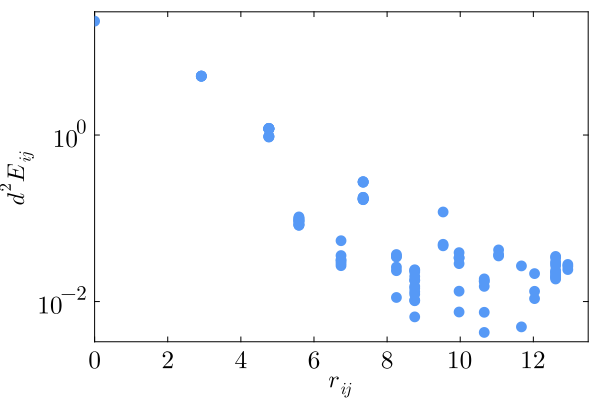}}
    \caption{$\rm Mg_2Si$ and $\rm C$: the strong locality in a homogeneous system depends on $\gap$ rather than $\Gap$; cf. \eqref{eq:strong}.}
    \label{fig:Mg2Si_strong}
\end{figure}

Finally, we examine the strong locality of $\rm{Mg_2Si}$ and $\rm C$ by evaluating the energy Hessian in homogeneous systems, which we computed using the finite difference of site forces $[F(\dmr)]_i, i \in \Lambda$, defined by
\begin{align}
   d^2E_{ij} := \left|\frac{\partial^2 E(\dmr)}{\partial \br_i\partial \br_j}\right| \approx \left|\frac{[F\big(\dmuEp\big)]_i - [F(\dmr)]_i}{\partial \br_j}\right|, 
\end{align}
where $u_{\ep}(\ell) = 0$ if $\ell \neq j$ and $ u(j)= \ep j$ for some $\ep >0$ as the step length.
In \Cref{fig:Mg2Si_strong}, by plotting the data $(r_{ij}, d^2E_{ij})$, we find it challenging to capture the decay in $\rm{Mg_2Si}$ due to its small $\gap$. In contrast, the decay is clearly visible in $\rm C$, which has a moderate $\gap$, comparable to $\Gap$.

\section{Conclusions}
\label{sec:Conclusions}
In this paper, we refined our understanding of the dependence of nearsightedness in linear tight binding models on the band structure. We focussed on materials with indirect band gaps, which are usually small in semiconductors such as ${\rm Si}$ and ${\rm Mg_2Si}$. In these systems, existing results suggest that small gaps theoretically render the locality negligible.

By performing the locality analysis in reciprocal space, we obtained a sharper estimate that highlights how the off-diagonal decay of the density matrix depends on $\Gap$, which is larger than $\gap$, as shown in \Cref{fig:gap_bz}. That is, in the small gap limit, the nearsightedness property can still be maintained as long as $\Gap \gg \gap$.
In this context, our results explain why some small gapped systems still exhibit fast decay of local operators, as shown in \Cref{fig:weak_loc} for the 1D toy model and \Cref{fig:Mg2Si_weak} for 3D systems. 

Moreover, we also demonstrated that the impact of perturbations decays as the lattice sites move far away from the perturbations, and this decay is controlled by $\gap$ and the norm of perturbations.
Regarding strong locality (e.g., interatomic interactions), our results, combined with those in \cite{ortner2020locality}, indicate that strong locality depends on $\gap$ and thus may lose fast decay in small gapped systems when considering the Hessian of energy or forces, as shown in \Cref{fig:Mg2Si_strong}. 
Moreover, the results in \Cref{fig: sqrt_Gap} and \Cref{fig:strong_loc} suggest that, under certain circumstances, the decay exhibits the square-root behaviour with respect to $\gap$ or $\Gap$, indicating the possibility of sharp estimates. To explore this further, we conducted a focused case study to support and justify these observations. 
A more detailed analysis of this effect will be pursued in future work.

\section*{Acknowledgements}
We gratefully acknowledge stimulating discussions with Antoine Levitt and Alexander Watson. 
HC's work was supported by the National Natural Science Foundation of China (No. NSFC12371431).
JF acknowledges the financial support from the China Scholarship Council (Grant No. 202306040154). 
CO was supported by NSERC Discovery Grant GR019381 and NFRF Exploration Grant GR022937. 
JT was supported by 
%
the Engineering and Physical Sciences Research Council (EPSRC) Grant EP/W522594/1. 
%

\section{Proofs of the Main Results}
\label{sec:Proofs of the Main Results}

\subsection{Proof of \texorpdfstring{\Cref{thm:nearsight}}{}}
\label{proof:thm_nearsight}
In preparation of the proof, we begin by proving a result on perturbations of the spectrum of $\Hxi$, given by \eqref{eq:block_ham}:

\begin{lemma}[Perturbations of the Spectrum]
\label{lem:pertb_spectrum}
    Fix $\bxi \in \Gamma^*$. Then, 
    \begin{align*}
        \sigma\left(\Ham_{\bxi + i\bzeta}\right) \subset B_{\ccRef|\bzeta|}\left(\sigma\big(\Hxi\big)\right),
        \quad 
        \forall \bzeta \in \R^d ~\text{with}~ |\bzeta|\leq\frac{1}{2}\gamma_0,
    \end{align*}
    where $\ccRef = c h_0 \gamma_0^{-(d+1)}$ with $c>0$ depending only on $d, \m, M$ and $\nob$. 
\end{lemma}

\begin{proof}
    According to \Cref{appd:bloch_trans}, for all $1\leq a,b \leq \nob$ and $\ell_0, k_0 \in \Lambda_0 = \Gamma\cap \Lambda$, we have 
    \begin{align*}
        \big[\Ham_{\bxi + i\bzeta}\big]_{\ell_0 k_0}^{ab} =
        \sum_{\alpha\in \mathbb Z^d} \Big([\Hr]_{\ell_0 + \mathsf{A}\alpha, k_0}^{ab} e^{(\ell_0 - k_0 +\mathsf{A}\alpha) \cdot \bzeta}  \Big)
        e^{- i (\ell_0 - k_0 +\mathsf{A}\alpha) \cdot \bxi}, 
    \end{align*}
    and this yields that
 \begin{align*}
    \begin{split}
        \left|\big[\Ham_{\bxi + i\bzeta}\big]_{\ell_0 k_0}^{ab} - \big[\Hxi\big]_{\ell_0 k_0}^{ab}\right|
        = &
        \left| \sum_{\alpha\in\Z^d}  [\Hr]_{\ell_0 + \mathsf{A}\alpha, k_0}^{ab}\big(e^{(\ell_0 - k_0 +\mathsf{A}\alpha) \cdot \bzeta} - 1\big)\right| \\
        \leq &
        h_0\sum_{\alpha\in\Z^d} e^{-\gamma_0|\ell_0 - k_0 +\mathsf{A}\alpha|} \big(e^{|\ell_0 - k_0 +\mathsf{A}\alpha||\bzeta|} - 1\big) .
    \end{split}
\end{align*}
For fixed $\ell_0, k_0$, we define
\begin{align*}
    F(\beta) = \sum_{\alpha\in\Z^d} e^{-\gamma_0|\ell_0 - k_0 +\mathsf{A}\alpha|} \big(e^{\beta|\ell_0 - k_0 +\mathsf{A}\alpha|} - 1\big),
\end{align*}
which is well-defined and differentiable on $\beta \in [0,\frac{1}{2}\gamma_0]$. Additionally, we see that $F(0) = 0$ and the derivative satisfies 
\begin{align*}
    F'(\beta) =
    \sum_{\alpha\in\Z^d} e^{-(\gamma_0 - \beta)|\ell_0 - k_0 +\mathsf{A}\alpha|}|\ell_0 - k_0 +\mathsf{A}\alpha|
    &\leq
    \sum_{\alpha\in\Z^d} e^{-\frac{1}{2}\gamma_0|\ell_0 - k_0 +\mathsf{A}\alpha|}|\ell_0 - k_0 +\mathsf{A}\alpha| \leq \frac{C_{\m, d}}{\gamma_0^{d + 1}}, 
\end{align*}
where $C_{\m, d}$ is a constant depending on $\m$ and $d$. In particular, we have $|F(\beta)|\leq \frac{C_{\m, d}}{\gamma_0^{d + 1}} \beta$.
Based on these, we can control the difference of $\Ham_{\bxi + i\bzeta}$ and $\Hxi$ by $|\bzeta|$ as follows 
\begin{align}
\label{eq:ham_k_diff}
    \begin{split}
        \left\|\Ham_{\bxi + i\bzeta} -  \Hxi\right\|_{\mathrm{F}}^2
        = & \sum_{\substack{\ell_0, k_0 \in \Lambda_0\\1\leq a,b \leq \nob }} \left|\big[\Ham_{\bxi + i\bzeta}\big]_{\ell_0 k_0}^{ab} - \big[\Hxi\big]_{\ell_0 k_0}^{ab}\right|^2\\
        \leq &
        (\nob h_0)^2\sum_{\substack{\ell_0, k_0 \in \Lambda_0 }} \left(\sum_{\alpha\in\Z^d} e^{-\gamma_0|\ell_0 - k_0 +\mathsf{A}\alpha|} \big(e^{|\ell_0 - k_0 +\mathsf{A}\alpha||\bzeta|} - 1\big) \right)^2 \\
        \leq &
        (M \nob h_0)^2\sup_{\substack{\ell_0, k_0 \in \Lambda_0}}\left(\sum_{\alpha\in\Z^d} e^{-\gamma_0|\ell_0 - k_0 +\mathsf{A}\alpha|} \big(e^{|\ell_0 - k_0 +\mathsf{A}\alpha||\bzeta|} - 1\big) \right)^2\\
        \leq &
        \left(C_{\m, d} M \nob h_0\gamma_0^{-(d +1)}\right)^2|\bzeta|^2
    \end{split}
\end{align}
for all $|\bzeta|\leq \frac{1}{2}\gamma_0$. By Weyl's theorem \cite{kato2013perturbation},  we have 
\begin{align}
\label{eq:ham_dist}
 \dist\left(\sigma\big(\Hxi\big),\, \sigma\big(\Ham_{\bxi + i\bzeta}\big)\right) \leq  
 \left\|\Ham_{\bxi + i\bzeta} -  \Hxi\right\|_{\mathrm{F}} \leq \ccRef |\bzeta|,
\end{align}
where $\ccRef := c h_0 \gamma_0^{-(d+1)}$ with $c > 0$ depending on $d, \m, M$ and $\nob$. Therefore, we can conclude that
\begin{align*}
    \sigma\left(\Ham_{\bxi + i\bzeta}\right) \subset B_{\ccRef |\bzeta|}\left(\sigma\big(\Hxi\big)\right),
\end{align*}
as required.
\end{proof}

Now we will use Lemma~\ref{lem:pertb_spectrum} to prove \Cref{thm:nearsight}:
\begin{proof}[Proof of \Cref{thm:nearsight}]
   Let us choose a simple, closed and positively oriented contour $\cc$ containing the line segment $i[-R,R]$ for some $R > 0$ (chosen later to be sufficiently large) and such that $\ep_{\mathrm F} + \mathscr C$ encircles the valence bands $\{\ep_n \}_{n \leq N_0}$ of $\Hr$ and satisfies
   \begin{align}
   \label{eq:contour_dist}
    \dist\Big(\ep_{\mathrm
     F} + \cc, \sigma(\Hr)\Big) > 0.
    \end{align}  
   Applying the contour integral representation in \eqref{eq:dm_k}, we have
    \begin{align}
    \label{eq:dm_component}
        \dmr_{\ell k, ab}
        &= \dashint_{\Gamma^*} \bigg[\oint_{\cc + \nu(\bxi)} (z - \Hxi)^{-1} \frac{\mathrm{d}z}{2\pi i} \bigg]_{\ell_0 k_0, ab} e^{i \bxi\cdot (\ell - k) }\mathrm{d}\bxi, 
    \end{align}
    where $\nu: \Gamma^* \rightarrow \C$ is a mapping to be determined,  satisfying $\varepsilon_{N_0}(\bxi) < \nu(\bxi) < \varepsilon_{N_0+1}(\bxi)$ for all $\bxi \in \Gamma^*$. 
    For simplicity, 
    let us define $R_z(\bxi):= \big( z - \Hxi\big)^{-1}$ and 
    \begin{align}
    \label{eq:res_int}
        \widehat{\rho}(\bxi) \coloneqq \oint_{\cc + \nu(\bxi)} (z - \Hxi)^{-1} \frac{\mathrm{d}z}{2\pi i} 
        = \oint_{\cc + \nu(\bxi)} R_z(\bxi)\, \frac{\mathrm{d}z}{2\pi i} 
        \quad\in \C^{\na\nob \times \na\nob}.
    \end{align}
    Then \eqref{eq:dm_component} can be rewritten as 
    \begin{align}
    \label{eq:dm_fourier}
        \dmr_{\ell k, ab}  = \dashint_{\Gamma^*}\big[\widehat{\rho}(\bxi)\big]_{\ell_0k_0, ab}\, e^{i \bxi\cdot (\ell - k) }\mathrm{d}\bxi,
    \end{align}
    thus the decay rate of $\big[\dmr_{\ell k, ab}\big]$ is related to the analyticity of $ \bm \xi \mapsto \left[\widehat{\rho}(\bxi)\right]_{\ell_0k_0, ab}$ \cite{reed1975ii}, which is determined by the regularity of $\nu(\cdot)$ in \eqref{eq:res_int}. 

    We first define
    $\mathsf{g}(\bm \xi) \coloneqq \ep_{N_0 + 1}(\bm \xi) - \ep_{N_0}(\bm \xi)$ and note that $\mathsf{gap}_+ = \min_{\Gamma^*} \mathsf{g}$. Next, introduce 
    $\alpha > 0$ to be determined later, and 
    \begin{align}
        \label{eq: D_R}
        D_{\mathrm{R}} \coloneqq \left\{ \bm \xi + i \bm \zeta \colon \bm \xi, \bm \zeta \in \mathbb R^d, |\bm \zeta| \leq \min\Big\{\tfrac{1}{2}\gamma_0,\, \tfrac{\alpha}{ \ccRef} \mathsf{gap}_+
        \Big\}\right\}
    \end{align}
    so that $\sigma\big(\Ham_{\bxi + i\bzeta}\big) \subset B_{\alpha \mathsf{gap}_+}\big( \sigma(\Hxi) \big)$ for $\bm \xi + i \bm \zeta \in D_{\mathrm
    R}$ with $\bm \xi, \bm \zeta \in \mathbb R^d$. The mapping $\nu_0 \colon \mathbb C^d \to \mathbb C$ defined by 
    \begin{align}
    \label{eq:def_nu_0}
        \nu_0(\bxi + i \bzeta) \coloneqq \frac{\ep_{N_0}(\bxi) + \ep_{N_0+1}(\bxi)}{2},
        \quad \bm \xi, \bm \zeta \in \mathbb R^d
    \end{align} 
    is continuous on $D_{\mathrm R}$ and hence admits a polynomial approximation $\nu \colon D_{\mathrm R} \to \mathbb C$ with 
    \begin{align}
    	\label{eq:nu_0}
        \big|\nu(\bxi + i\bzeta) - \nu_0(\bxi + i\bzeta)\big| < \frac{\alpha}{2}
        \mathsf{gap}_+, \quad \bxi + i\bzeta \in D_{\rm R}.
    \end{align}
    We therefore have $\nu(\bm \xi + i\bm \zeta) \in B_{\frac{\alpha}{2} \mathsf{gap}_+}\big( \frac{1}{2}(\ep_{N_0}(\bm \xi) + \ep_{N_0 + 1}(\bm \xi)) \big)$ and thus 
    \begin{align}
    \label{eq:distz_nu}
    \mathrm{dist}\left( \nu(\bm \xi + i \bm \zeta), \sigma\big(\Ham_{\bxi + i \bzeta} \big) \right) 
    \geq 
    \frac
        {\mathsf{g}(\bxi) - 3 \alpha \mathsf{gap}_+}
        {2}
    \geq \frac{1-3\alpha}{2} \mathsf{g}(\bm \xi).
    \end{align}
    We choose $\alpha \coloneqq \frac{1}{6}$ and $R$ sufficiently large (recall, $i[-R,R] \subset \mathscr C$) so that $\dist\left(z, \sigma\big(\Ham_{\bxi + i \bzeta} \big)\right) \geq \frac{1}{4} \mathsf{g}(\bm \xi)$ for all $z \in \nu(\bm \xi) + \mathscr C$ and $\bxi + i\bzeta \in D_{\rm R}$. 
    Since $\nu \colon \mathbb C^d \to \mathbb C$ is analytic on $D_{\mathrm{R}}$, we have that $\widehat{\rho}$ extends to an analytic function on $D_{\mathrm{R}}$. 
    
    Therefore, for all $|\bm \zeta| \leq \eta_+ \coloneqq \min\{ \frac{1}{2}\gamma_0, \frac{\alpha}{c_{\mathrm{r}}} \mathsf{gap}_+\}$, we have
    \begin{align}
    \label{eq:dm_decay_prefac}
        C^2\coloneqq \int_{\Gamma^\star} |\widehat{\rho}(\bm \xi + i \bm \zeta)_{\ell_0 k_0, ab}|^2 \mathrm{d}\bm \xi
        = \sum_{k} \left|\dmr_{\ell k, ab}\right|^2 e^{2\,\bzeta \cdot (\ell - k)}  < \infty.
    \end{align}
    That is, we have the required bound $\big|\dmr_{\ell k, ab}\big| \leq C e^{-\etaP \,r_{\ell k}}$.

    Finally, we may show the dependence of $C$ by noticing that, 
    \begin{align}
    \label{eq:const}
        C^2 
        \leq 
        \int_{\Gamma^\star}
            \left( \frac
            {\ell(\mathscr C)}
            {2\pi}
            \sup_{z \in \nu(\bm \xi+i\bm \zeta) + \mathscr C}|R_z(\bm \xi+i\bm \zeta)_{\ell_0 k_0, ab}|
        \right)^2
        \mathrm{d}\bm \xi
        \leq  
        \left( 
        \frac
            {2\ell(\mathscr C)}
            {\pi}
        \right)^2
        \int_{\Gamma^*} \mathsf{g}(\bm \xi)^{-2} \dd\bxi
    \end{align}
    where $\ell(\mathscr C)$ denotes the length of the contour. Here, we have used the fact that $|R_z(\bm \xi)_{\ell_0 k_0,ab}| \leq \dist\big(z,\sigma(\Hxi)\big)^{-1} \leq 4 \mathsf{g}(\bm \xi)^{-1}$. This concludes the estimate \eqref{eq:thm_nearineq}.
\end{proof}

\subsection{Proof of \texorpdfstring{\Cref{thm:prtb}}{}}
\label{proof:prtb}
In order to prove the conclusion, we require some preliminary results from \cite{ortner2020locality}. Here we will restate these conclusions in our setting, and give some variants to establish a foundation for the subsequent proof. 

We start with an improved Combes--Thomas resolvent estimate: 

\begin{lemma}
    \label{lem:CT_est_2}
     Suppose $u \in \adm(\Lambda)$ and $z\in \mathbb C$ such that $\dist\left(z, \sigma\big(\Hr\big)\right) \geq \Dr$ and $\dist\left(z, \sigma\big(\Ham(u)\big)\right) \geq \D(u)$ for some $\Dr,\D(u) > 0$. 
     Then, 
    \begin{align}
    \label{eq:CT_est_2}
        \Big|\big(z - \Ham(u)\big)^{-1}_{\ell k, ab}\Big|
            \leq 
            \frac{4}{\Dr}e^{-\CT(\frac{1}{2}\Dr) r_{\ell k}} + 
            C(u) e^{-\frac{1}{2}\CT(\frac{1}{2}\Dr)(|\ell| + |k|)},
    \end{align}
    with $\gamma_{\mathrm{CT}}(\frac{1}{2}\Dr) = c_0 \min\{ 1, \frac{1}{2}\Dr\}$, and 
    \begin{align}
        \label{eq:cu}
        C(u) \coloneqq C 
            \frac
            {e^{\CT(\frac{1}{2}\Dr) R_1}}
            {(\D^{\rm ref})^{d + 4} }
            \frac
                {\|Du\|_{\ell^2_{\Upsilon}}}
                {\D(u)^{d+1}},
        \quad 
        R_1>0 \text{ s.t. } \|Du\|_{\ell^2_{\Upsilon}(\Lambda \setminus B_{R_1})} \leq c_1 \mathfrak{d}^{\rm ref}
    \end{align}
    for some $c_1>0$, and $c_0, C > 0$ depend only on $d, \mathfrak{m}, h_0, \gamma_0, N_{\mathrm b}$ and $M$. In particular, for sufficiently small $\|Du\|_{\ell^2_{\Upsilon}}$, $R_1$ can be chosen arbitrarily. 
\end{lemma}

\begin{proof}
    The point is mainly based on the approximation of $\Ham(u)$ by the finite rank update of $\Hr$~\cite{ortner2020locality}. We summarise the proof here, as parts of the ideas will be revisited in the subsequent discussion.  \cite[Lemma 5.4]{ortner2020locality} states that, for each $\delta >0$, there exist $R_{\delta} >0$ (depending on $\delta$) and operators $P_{\delta}(u), P_{\loc}(u)$ such that 
    \begin{align}
    \label{eq: finite_update}
        \Ham(u) = \Hr + P_{\delta}(u) + P_{\loc}(u),
    \end{align} 
    where $\|P_{\delta}(u)\|_{\mathrm{F}} \leq \delta$, and $[P_{\loc}(u)]^{ab}_{\ell k} = [\Ham(u) - \Ham^{\rm ref}]_{\ell k}^{ab}$ if $(\ell, k) \in B_{R_{\delta}} \times B_{R_{\delta}}$; else, $[P_{\loc}(u)]^{ab}_{\ell k} = 0$.
    That is, we can decompose $\Ham(u) - \Ham^{\rm ref}$ into two parts: $P_{\delta}$ and $P_{\rm loc}$, which are small in terms of Frobenius norm and rank, respectively. As for the selection of $R_{\delta}$, 
    by \cite[Lemma 5.4]{ortner2020locality}, we can first choose a sufficiently large $R_1>0$ satisfying $|u(\ell) - u(k)| \leq \m |\ell - k|$ for all $\ell, k \in \Lambda \setminus B_{R_1}$, so that 
\begin{align}
\label{eq:ham_diff_mid}
    \Big|\Ham(u)_{\ell k, ab} - \Hr_{\ell k, ab}\Big|
        &\leq
        h_0 e^{-\frac{\sqrt{3}}{2}\gamma_0\m r_{\ell k}}
        |u(\ell) - u(k)|
        \quad \forall \ell, k \in \Lambda\setminus B_{R_1}, 
\end{align}
and then, for $R_2 > 0$, the regularity of $\Ham(u)$ and $\Hr$ gives
\begin{align}
\label{eq:ham_diff_out}
    \sum_{\ell \in B_{R_1}} \sum_{\substack{k\in \Lambda\backslash B_{R_1}\\{|\ell -k|> R_2}}}
    \Big|\Ham(u)_{\ell k, ab} - \Hr_{\ell k, ab}\Big|^2
        \leq
        C\sum_{\ell \in B_{R_1}} \sum_{\substack{k\in \Lambda\backslash B_{R_1}\\{|\ell -k|> R_2}}} e^{-2\gamma_0\m|\ell -k|}
        \leq CR_1^d e^{-\gamma_0\m R_2}, 
\end{align}
where $C>0$ depends on $d, \m, \gamma_0, h_0$. Through $\eqref{eq:ham_diff_mid}$ and $\eqref{eq:ham_diff_out}$, we have
\begin{align}
\label{eq:p_delta}
\|P_{\delta}\|^2_{\rm F}
    &=
    \sum_{1\leq a, b \leq \nob}
    \Big(
        \sum_{\ell, k \in \Lambda\backslash B_{R_1}} + \sum_{\ell \in B_{R_1}} \sum_{\substack{k\in \Lambda\backslash B_{R_1}\\{|\ell -k|> R_2}}}
    \Big)
    \Big|\Ham(u)_{\ell k, ab} - \Hr_{\ell k, ab}\Big|^{2} \nonumber\\
    &\leq
    C\big(
        \|Du\|_{\ell^2_{\Upsilon}(\Lambda\backslash B_{R_1})}^{2} 
        + R_1^d e^{-\gamma_0\m R_2}
    \big), 
\end{align}
which implies the value can be arbitrarily small by choosing $R_1$ and then $R_2$ sufficiently large.
Specifically, we can choose $R_1, R_2$ with
\begin{align}
    \label{eq:p_Dref}
    C\big(
        \|Du\|_{\ell^2_{\Upsilon}(\Lambda\backslash B_{R_1})}^{2} 
        + R_1^d e^{-\gamma_0\m R_2}
        \big)
        \leq \big(\frac{1}{2}\Dr\big)^2,
\end{align}
By setting $R_{\delta} = R_1 + R_2$, we obtain
$\|P_{\delta}\|_{\rm F} \leq \frac{1}{2}\Dr$ and $P_{\rm loc}$ satisfying the identity \eqref{eq: finite_update}, where the choice of $R_{\delta}$ depends on $\Dr$.

Since $P_{\loc}(u)$ has finite rank, we can perform a decomposition for the submatrix $P_{\loc}|_{B_{R_{\delta}} \times B_{R_{\delta}}} := \{(P_{\loc})_{ij}\}_{(i,j) \in B_{R_{\delta}} \times B_{R_{\delta}}}$, by operators $U, V$ satisfying $P_{\loc} = UV$, where $U|_{B_{R_{\delta}} \times B_{R_{\delta}}}$ is orthogonal and $U_{\ell i} = 0, V_{j\ell} = 0$, for $\forall i,j \in B_{R_{\delta}}$ and $\ell \in \Lambda \setminus B_{R_{\delta}}$. For brevity, we denote $A := z - \Hr - P_{\delta}(u)$, then $A - P_{\rm loc}(u) = z - \Ham(u)$.
Applying the Woodbury identity~\cite{hager1989updating} for $A - P_{\rm loc}(u)$ yields
\begin{align}
\label{eq:ham_finite_prtb}
    \big(A - P_{\rm loc}(u)\big)^{-1}
        = \big(A - UV\big)^{-1}
        = A^{-1} + A^{-1}U\big(I - VA^{-1}U\big)^{-1}VA^{-1},
\end{align}
where $I$ is the unit matrix of size $|B_{R_{\delta}}| \times |B_{R_{\delta}}|$.
Notice that $U\big(I - VA^{-1}U\big)^{-1}V$ also has finite rank. 
Thus, for $\forall \ell, k \in B_{R_{\delta}}, 1 \leq a, b \leq \nob$, by following \cite{ortner2020locality}, we have 
\begin{align}
\label{eq:URV}
    \begin{split}
     &\quad   
     \left|\left[U\big(I - VA^{-1}U\big)^{-1}V\right] _{\ell k, ab} \right| 
         = 
         \left|\left[
         A\big(A - P_{\rm loc}(u)\big)^{-1}P_{\rm loc}(u)
         \right] _{\ell k, ab} 
         \right| \\
     &\leq 
         \sum_{\substack{m,n:|\br_n|\leq R_{\delta}}} 
         \sum_{a_1 b_1} 
         \Big|\big(z - \Hr - P_{\delta}(u)\big)_{\ell m , a a_1}\Big|
         \left|\big(z - \Ham(u)\big)^{-1}_{mn,a_1b_1} \right| 
         \big|P_{\rm loc}(u)_{nk, b_1 b}\big|\\
     &\leq 
         \frac{2\nob (|z| + h_0)}{\D(u)} 
         \left(
            \sum_{\substack{m,n:|\br_n|\leq R_{\delta}}} 
            e^{-\gamma_0r_{\ell m}} e^{-\CT(\D(u)) r_{mn}}
         \right)  
         \max_{\substack{\ell, k \in B_{R_{\delta}}\\ 1\leq a,b \leq \nob}} \big|P_{\mathrm{loc}}(u)_{\ell k, ab}\big|\\
     & \leq 
        \frac{c}{\big(\D(u)\big)^{d+1}}
        \max_{\substack{\ell, k \in B_{R_{\delta}}\\ 1\leq a,b \leq \nob}} \big|P_{\mathrm{loc}}(u)_{\ell k, ab}\big|,
    \end{split}
\end{align}
where $c > 0$ depends on $h_0, \gamma_0, \m, d, \nob, M$, and is bounded independently of $z$ since we are considering $z$ in a bounded set. From the above proof of \eqref{eq: finite_update}, we can control the maximum of $P_{\rm loc}(u)$ in $B_{R_{\delta}}$ by
\begin{align}
\label{eq:P_loc}
    \max_{\substack{\ell, k \in B_{R_{\delta}}\\ 1\leq a,b \leq \nob}} |P_{\mathrm{loc}}(u)|
        &= 
            \max_{\substack{\ell, k \in B_{R_{\delta}}\\ 1\leq a,b \leq \nob}}
            \left| \big[\Ham(u) - \Hr - P_{\delta}\big]_{\ell k}^{ab} \right|
            =
            \max_{\substack{\ell, k \in B_{R_{\delta}}\\ 1\leq a,b \leq \nob}}
            \left| \big[\Ham(u) - \Hr \big]_{\ell k}^{ab} \right|
            \nonumber\\
        &\leq
            \max_{\substack{\ell, k \in B_{R_{\delta}}\\ 1\leq a,b \leq \nob}}
            \Big|
                \nabla h^{ab}_{\ell k}\big( \pmb{\nu}_{\theta} \big)\cdot 
                \big[u(\ell) - u(k)\big]
            \Big|
        \leq
            C\|Du\|_{\ell^2_{\Upsilon}(B_{R_{\delta}})}, 
\end{align}
 where we applied mean value theorem and $\pmb{\nu}_{\theta}:=  (1 - \theta)\br_{\ell k}(u) + \theta\br_{\ell k}$ for some $\theta \in [0, 1]$, and the constant $C>0$ depends only on $d, \m, h_0, \gamma_0$. Additionally, as $\|P_{\delta}\|_{\rm F} < \frac{1}{2}\Dr$, the Combes–Thomas estimate \eqref{eq:CombesThomas} for $A$ is
\[
\big|A^{-1}_{\ell k, ab}\big|
    \leq 
    \frac{4}{\Dr}
    e^{-\CT(\frac{1}{2}\mathfrak{d}^{\mathrm{ref}}) r_{\ell k}}.
\]
Combining this with \eqref{eq:URV}, \eqref{eq:P_loc} and the exponential sum estimate in Lemma~\ref{lem:expsum}, the second term of \eqref{eq:ham_finite_prtb} is bounded by
\begin{align}
    \left|
    \left[A^{-1}U\big(I - VA^{-1}U\big)^{-1}VA^{-1}\right] _{\ell k, ab} \right| 
        & \leq
            \frac{c\|Du\|_{\ell^2_{\Upsilon}(B_{R_{\delta}})}}{(\D^{\rm ref})^2 \big(\D(u)\big)^{d+1}}
            \sum_{\ell_1,\ell_2 \colon |\ell_l| \leq R_{\delta}}
                e^{-\CT(\frac{1}{2}\D^{\rm ref}) (r_{\ell\ell_1}+ r_{\ell_2 k})} \nonumber\\
        &\leq
            \frac{c\|Du\|_{\ell^2_{\Upsilon}}}{(\D^{\rm ref})^2 \big(\D(u)\big)^{d+1}}
            \left(
                \frac
                    {C_{d} e^{\frac{1}{2}\CT(\frac{1}{2}\Dr) R_{\delta}}}
                    {(\m\CT(\frac{1}{2}\Dr))^d} 
                e^{-\frac{1}{2}\CT(\frac{1}{2}\D^{\rm ref})(|\ell| + |k|)}
            \right)
            \nonumber\\
        & \leq
            C(u)\,e^{-\frac{1}{2}\CT(\frac{1}{2}\D^{\rm ref})(|\ell| + |k|)},\label{eq:u_depend}
\end{align}
where $C(u) := Ce^{\frac{1}{2}\CT(\frac{1}{2}\Dr) R_{\delta}} \|Du\|_{\ell^2_{\Upsilon}}\big/\big((\D^{\rm ref})^{d+2} \big(\D(u)\big)^{d+1} \big)$ with $C$ depending on $h_0, \gamma_0, \m, d, \nob$ and $M$. Noticing that, from \eqref{eq:p_Dref}, we can find some constants $C_1, C_2>0$ such that $e^{\frac{1}{2}\CT(\frac{1}{2}\Dr) R_2} = C_1 R_1^d(\Dr)^{-2}$, and $\|Du\|_{\ell^2_{\Upsilon}(\Lambda\setminus B_{R_1})} < C_2\Dr$. Hence, the prefactor $C(u)$ can be written as 
\begin{align*}
    C(u) = C 
            \frac
            {e^{\CT(\frac{1}{2}\Dr) R_1}}
            {(\D^{\rm ref})^{d + 4} }
            \frac
                {\|Du\|_{\ell^2_{\Upsilon}}}
                {\D(u)^{d+1}}.
\end{align*}
Finally, taking \eqref{eq:CombesThomas} and \eqref{eq:u_depend} into the entry of \eqref{eq:ham_finite_prtb}, we conclude that
\begin{align}
\label{eq:improved_CT}
    \begin{split}
        \left|\big(z - \Ham(u)\big)^{-1}_{\ell k, ab} \right|
            & \leq
                \Big|A^{-1}_{\ell k, ab}\Big| + 
                \left|
                    \big[A^{-1}U\big(I - VA^{-1}U\big)^{-1}VA^{-1}\big] _{\ell k, ab} 
                \right| \\
            &\leq
                \frac{4}{\Dr}
                e^{-\CT(\frac{1}{2}\D^{\rm ref}) r_{\ell k}} + 
                C(u) e^{-\frac{1}{2}\CT(\frac{1}{2}\D^{\rm ref})(|\ell| + |k|)}.
    \end{split}
\end{align}
\end{proof}

\begin{remark}
By directly applying Lemma~\ref{lem:CT_est_2}, we have
\begin{align}
    \label{eq:combes_improv}
    \left|\big(z - \Ham(u)\big)^{-1}_{\ell k, ab} \right|
        \leq &
            \left(
                \frac{4}{\Dr}e^{-\frac{1}{2}\CT(\frac{1}{2}\D^{\rm ref})r_{\ell k}} + C(u)e^{-\frac{1}{2}\CT(\frac{1}{2}\D^{\rm ref})
                	(|\ell| + |k| - r_{\ell k})}
            \right)
            e^{-\frac{1}{2}\CT(\frac{1}{2}\D^{\rm ref})r_{\ell k}}
            \nonumber \\
        := &
            C_{\ell k}(u)
            e^{-\frac{1}{2}\CT(\frac{1}{2}\D^{\rm ref})r_{\ell k}},
\end{align}
where $C_{\ell k}(u)$ satisfies $C_{\ell k}(u) \lesssim \frac{1}{\Dr}$ as $|\ell| + |k| - r_{\ell k} \to \infty$.
Additionally, we can always choose $\cc$ in the above estimation satisfying 
\[
    \dist(\cc, \sigma(\Hr))\geq \frac{1}{2}\gap,
\] 
so that the above estimate can be written in terms of $\gap$ and $\CT(\frac{1}{2}\gap)$.
\end{remark}

\vspace{0.3cm}
With the above preparation, we now proceed to the proof of \Cref{thm:prtb}:
\begin{proof}[Proof of \Cref{thm:prtb}]
Notice that the following inequality always holds 
\begin{align}
    \label{eq:tri_ineqt}
    \big|\dmu_{\ell k, ab} \big|
    \leq
    \big|\dmr_{\ell k, ab}\big|  
    +  \big|\dmu_{\ell k, ab}  - \dmr_{\ell k, ab} \big|
    \quad \forall \ell, k \in \Lambda, 
    1 \leq a,b \leq \nob.
\end{align}
To obtain the result, it suffices to only estimate the last term in \eqref{eq:tri_ineqt}, for the control of the first term has been given by \Cref{thm:nearsight}.
Applying the finite rank update \eqref{eq: finite_update}, 
the difference between $\dmu$ and $\dmr$ is bounded by
\begin{align}
\label{eq:dm_pertb_diff}
        \quad \left|\dmu_{\ell k, ab}  - \dmr_{\ell k, ab}\right|
        &= \left|\oint_{\cc} \Big[\big(z - \Ham(u)\big)^{-1}- \big(z - \Hr\big)^{-1}\Big]_{\ell k, ab} \frac{\dd z}{2\pi i }\right| \nonumber\\
        &= \left| \oint_{\cc} \left[ \big(z - \Ham(u)\big)^{-1}\big(\Ham(u) - \Hr\big)\big(z - \Hr\big)^{-1}\right]_{\ell k, ab}
        \frac{\dd z}{2\pi i } \right|\nonumber\\
        & \leq \sum_{mn} \oint_{\mathscr C}\Big| (z - \Ham(u))^{-1}_{\ell m} ( P_\delta(u) + P_{\rm loc}(u) )_{mn} (z - \Ham^{\rm ref})^{-1}_{nk}  \Big|       \frac{\dd z}{2\pi i }
        \nonumber\\
        & \leq C\sum_{mn} \left| (z - \Ham(u))^{-1}_{\ell m} \right|
        \left(\left|
         P_\delta(u)_{mn} \right| +    \left|P_{\rm loc}(u) )_{mn}
        \right|
        \right)
        \left|(z - \Ham^{\rm ref})^{-1}_{nk} \right|  
        \nonumber\\
        &\eqqcolon C\left(T_{\delta} + T_{\rm loc}\right).
\end{align}
Recall that $[P_{\loc}(u)]^{ab}_{\ell k} = [\Ham(u) - \Ham^{\rm ref}]_{\ell k}^{ab}$ if $(\ell, k) \in B_{R_{\delta}} \times B_{R_{\delta}}$, and $[P_{\rm loc}]_{\ell k}^{ab} = 0$ if $\ell$ or $k \notin B_{R_{\delta}}$.
Here, 
 we simplify $R_{\delta}(= R_1 + R_2)$ as $R$ in \eqref{eq:p_delta}, by first choosing $R_1$ and then taking $R_2$ sufficiently large. Using \eqref{eq:ham_diff_mid} together with H\"older's inequality, we obtain
\begin{align}
\label{eq:T_delta}
    &T_\delta
        \coloneqq \sum_{mn} \left| (z - \Ham(u))^{-1}_{\ell m} \right|
        \left|
             P_\delta(u)_{mn} 
        \right| 
        \left|(z - \Ham^{\rm ref})^{-1}_{nk} \right| \nonumber\\ 
        &\leq
            C_{\ell k}(u)
            \sum_{m, n \in \Lambda\setminus B_R}
             e^{-\frac{1}{2}\CT(\frac{1}{2}\D^{\rm ref}) r_{\ell m}} 
            \Big(
                e^{-\frac{\sqrt{3}}{2}\gamma_0\m r_{mn}} |u(m) - u(n)|
            \Big)
            e^{-\CT(\D^{\rm ref})r_{nk}} \nonumber\\
         &\leq
            C_{\ell k}(u)
            \Big(
                \sum_{m, n \in \Lambda \backslash B_R} e^{-2 \Upsilon r_{m n}}|u(m)-u(n)|^2
            \Big)^{\frac{1}{2}} 
            \Big(
                \sum_{m, n \in \Lambda\setminus B_R}
                e^{-\CT(\frac{1}{2}\D^{\rm ref}) r_{\ell m}} 
                e^{-(\sqrt{3}\gamma_0\m - 2\Upsilon) r_{mn}}
                e^{-2\CT(\D^{\rm ref})r_{nk}}
            \Big)^{\frac{1}{2}} 
            \nonumber\\
        &\leq
             C_{\ell k}(u)\|Du\|_{\ell^2_{\Upsilon}(\Lambda\backslash B_R)}
             \Big(
                \sum_{m, n \in \Lambda\setminus B_R}
                e^{-\CT(\frac{1}{2}\D^{\rm ref}) r_{\ell m}} 
                e^{-\frac{\sqrt{3}}{2}\gamma_0\m r_{mn}}
                e^{-2\CT(\D^{\rm ref})r_{nk}}
            \Big)^{\frac{1}{2}} 
            \nonumber\\
        &\leq
             C_{\ell k}(u)\|Du\|_{\ell^2_{\Upsilon}(\Lambda\backslash B_R)} 
             e^{-\etaM r_{\ell k}},
\end{align}
where we denote
$
    \etaM :=
    \frac{1}{4}\min\{\CT(\tfrac{1}{2}\Dr), \gamma_0\m \} 
    = c_2\frac{\gamma_0}{h_0} \min\{h_0, \gamma_0^d\Dr\} 
$
with $c_2>0$ depending only on $d, \m, M$ and $\nob$.
Furthermore, by combining the exponential sum  estimate \eqref{eq:expsum-2} with Combes--Thomas estimate \eqref{eq:CombesThomas} and its improved form \eqref{eq:combes_improv}, we obtain
\begin{align}
\label{eq:T_loc}
 T_{\rm loc} 
    & \coloneqq \sum_{mn} \left| (z - \Ham(u))^{-1}_{\ell m} \right|
        \big|
             P_{\rm loc}(u)_{mn} 
        \big| 
        \left|(z - \Ham^{\rm ref})^{-1}_{nk} \right| 
        \nonumber\\
    & \leq
         C_{\ell k}(u)
         \|Du\|_{\ell^2_{\Upsilon}(B_{R})}
         \left(
            \sum_{m \in B_{R}} 
            e^{-\frac{1}{2}\CT(\frac{1}{2}\D^{\rm ref})  r_{\ell m}}
         \right)
         \left(
            \sum_{n \in B_{R}} e^{-\CT(\D^{\rm ref}) r_{nk}}
         \right)
         \nonumber\\
    & \leq
        C_{\ell k}(u) 
        \|Du\|_{\ell^2_{\Upsilon}(B_{R})}
        e^{\etaM R}
        e^{-\etaM (|\ell| + |k|)}
\end{align}
From Lemma~\ref{lem:expsum}, we note that the prefactors in \eqref{eq:T_delta} and \eqref{eq:T_loc} are both independent of $\etaM$ as $\gap$ is small. Moreover, by choosing $\cc$ satisfying $\dist(\cc, \sigma(\Hr))\geq \frac{1}{2}\gap$, we can take $\gap$ in $\etaM$.
Substituting \eqref{eq:T_delta} and \eqref{eq:T_loc} into \eqref{eq:dm_pertb_diff} draws the conclusion that 
\begin{align}
\label{eq:dm_diff}
    \big|\dmu_{\ell k, ab}  - \dmr_{\ell k, ab}\big|
        &\leq C(T_\delta + T_{\rm loc} ) 
            \nonumber\\
        &\hspace{-1.0cm}\leq
            C_{\ell k}(u)
            \Big(
                \|Du\|_{\ell^2_{\Upsilon}(\Lambda\backslash B_{R})}
                +
                \|Du\|_{\ell^2_{\Upsilon}(B_{R})}
                e^{\etaM R}e^{-\etaM(|\ell| + |k| - r_{\ell k})}
            \Big)
            e^{-\etaM r_{\ell k}}
            \nonumber\\
        &\hspace{-1.0cm}
        \coloneqq \tilde{C}_{\ell k}(u) e^{-\etaM r_{\ell k}}.
\end{align}
Notice that $R = R_1 + R_2$, and that the factor $C_{\ell k}(u)$ defined in \eqref{eq:cu} depends only on $R_1$. As $|\ell| + |k| - r_{\ell k} \to \infty$, we have $\tilde{C}_{\ell k}(u) \lesssim C\|Du\|_{\ell^2_{\Upsilon}(\Lambda \setminus B_{R})}$. Moreover, as $\|Du\|_{\ell^2_{\Upsilon}} \to 0$, taking the limit $R_2 \to \infty$ yields $\tilde{C}_{\ell k}(u) \to 0$. Finally, substituting \eqref{eq:dm_diff} into \eqref{eq:tri_ineqt} completes the proof.
\end{proof}

\subsection{Proof of the Theorem~\ref{thm: sharp_estimate}}
\label{proof: sqrt_pf}

Under assumption $\asReg$, 
we apply the mean value theorem to the bands $\ep_n(\bxi+i\bzeta), ~n \in\{ N_0, N_0 + 1\}$. Noting that $\nabla \ep_n(\bxi)$ is real, the first-order term $\re(i\bzeta^{\top}\nabla \ep_n(\bxi))$ vanishes, yielding: for all $\bm \xi, \bm \zeta$ such that $\bm \xi + i \bm \zeta \in \C_{\reg}^d$, we have
\begin{align}
	\label{eq: analy_exp}
	\re\ep_{n}(\bxi + i \bzeta) - \ep_{n}(\bxi) 
    & = 
    - \frac{1}{2} \bzeta^{\top} \nabla^2 \ep_n (\bxi') \bzeta,
\end{align}
for some $\bxi' \in [\bxi,\bxi +  i\bzeta]$.
First, we estimate the norm of the Hamiltonian in the complex domain.
For any $\bxi + i \bzeta \in \C_{\reg}^d$, we have 
\begin{align}
    \begin{split}
        \left\|\Ham_{\bxi + i\bzeta}\right\|_{\mathrm{F}}^2
        = & \sum_{\substack{\ell_0, k_0 \in \Lambda_0\\1\leq a,b \leq \nob }} \left|\big[\Ham_{\bxi + i\bzeta}\big]_{\ell_0 k_0}^{ab} \right|^2
        \leq 
        (\nob h_0)^2\sum_{\substack{\ell_0, k_0 \in \Lambda_0 }} \left(\sum_{\alpha\in\Z^d} e^{-\gamma_0|\ell_0 - k_0 +\mathsf{A}\alpha|} e^{|\ell_0 - k_0 +\mathsf{A}\alpha||\bzeta|} \right)^2 \\
        \leq &
        (M \nob h_0)^2\sup_{\substack{\ell_0, k_0 \in \Lambda_0}}\left(\sum_{\alpha\in\Z^d} e^{-(\gamma_0 - |\bzeta| ) |\ell_0 - k_0 +\mathsf{A}\alpha|}
         \right)^2\\
        \leq &
        \left(C_{\m, d} M \nob h_0(\gamma_0 - |\bzeta|)^{-d}\right)^2.
    \end{split}
\end{align}
By choosing an appropriate radius specifically for $|\bzeta| \leq \tfrac{1}{2} \min\{\gamma_0, \reg\}$,  we have $ \left\|\Ham_{\bxi + i\bzeta}\right\|_{\mathrm{F}}<\infty$. 
Furthermore, 
from \eqref{eq: analy_exp}, for $n \in \{N_0, N_0 + 1\}$, it follows that
\begin{align}
    \label{eq: real_ptb}
    	|\re\ep_{n}(\bxi + i \bzeta) - \ep_{n}(\bxi) | 
        \leq
        \frac{1}{2} \sup_{{\bm \omega} \in \C_{\reg/2}^d \cap \C_{\gamma_0/2}^d}
 		\|\nabla^2\ep_{n}({\bm \omega})\|_{\rm F}  |\bzeta|^2    
  		\leq 
  		m_{\infty}|\bzeta|^2. 		
\end{align}  

For any $\bxi, \bzeta \in \R^d$ and $\ell, k \in \Lambda$, we denote 
$\begin{psmallmatrix}
    \bm \xi \\ \bm \zeta 
\end{psmallmatrix} : \begin{psmallmatrix}
    \ell \\ k
\end{psmallmatrix} \coloneqq \bm \xi \cdot \ell + \bm \zeta \cdot k$
for brevity. 
Following the proof of Theorem~\ref{thm:nearsight}, we 
choose a suitable contour $\cc \supset i[-R, R]$ for some $R > 0 $ to be determined later, and write
\begin{align}
\label{eq: deriv_contour_integral}
\frac{\partial\rho_{\ell\ell}}{\partial [\bm r_k]_{j}}
& \coloneqq
\frac{\partial\big[\dmu\big]_{\ell\ell}}{\partial [u(k)]_{j}}\Bigg|_{u = \mathbf{0}} \nonumber\\
    &= 
    \oint_{\cc+\ep_{\rm F}} \left[\big(z - \Ham^{\rm ref}\big)^{-1} \,             \frac{\partial \Ham(u)}{\partial [u(k)]_{j}}\bigg|_{u = \mathbf{0}}\, \big(z - \Ham^{\rm ref}\big)^{-1} \right]_{\ell\ell}
    \frac{\dd z}{2\pi i } \nonumber\\
    &= \sum_{mn}
        \frac
            {\partial \Ham_{mn}}
            {\partial [\bm r_k]_{j}}
        \dashint_{(\Gamma^*)^2} 
            \left[ \oint_{\cc + \nu(\bxi_1, \bxi_2)}
                \big(z - \Ham_{\bxi_1}\big)^{-1}_{\ell_0 m_0} 
                \big(z - \Ham_{\bxi_2}\big)^{-1}_{n_0\ell_0}
                \frac{\dd z}{2\pi i }
            \right]
            e^{i \binom{\bxi_1}{\bxi_2} : \binom{\ell - m}{n - \ell}}
                \dd\bxi_1 \dd\bxi_2  \nonumber\\
    &\eqqcolon \sum_{mn}
        \frac
            {\partial \Ham_{mn}}
            {\partial [\bm r_k]_{j}}
        \dashint_{(\Gamma^*)^2} 
            \widehat{G}(\bxi_1, \bxi_2)_{\ell_0m_0,n_0\ell_0}
                e^{i \binom{\bxi_1}{\bxi_2} :  \binom{\ell - m}{n - \ell}}
                \dd\bxi_1 \dd\bxi_2 \nonumber\\
    &\eqqcolon \sum_{mn}
        \frac
            {\partial \Ham_{mn}}
            {\partial [\bm r_k]_{j}}
        \, G_{\ell m, n\ell}, 
\end{align}
For any $1\leq a, b \leq \nob$, we write $\frac{\partial\rho_{\ell\ell, ab}}{\partial [\bm r_k]_{j}} = \sum_{mn, a'b'}
        \frac
            {\partial \Ham_{mn, a'b'}}
            {\partial [\bm r_k]_{j}}
        \, G_{\ell m, n\ell}^{aa', bb'}$.
The decay properties of \eqref{eq: deriv_contour_integral} relies on the analyticity of the product function defined by
\begin{align}
\label{eq: res_time_res}
\widehat{G}(\bxi_1, \bxi_2)
    \coloneqq
        \oint_{\cc + \nu(\bxi_1, \bxi_2)}
            \big(z - \Ham_{\bxi_1}\big)^{-1}
            \otimes
            \big(z - \Ham_{\bxi_2}\big)^{-1}
            \frac{\dd z}{2\pi i }, 
        \quad
        (\bxi_1, \bxi_2) \in (\Gamma^*)^2,
\end{align}
where the contour shift satisfies $\max\{\ep_{N_0}(\bxi_1), \ep_{N_0}(\bxi_2)\} <\nu(\bxi_1, \bxi_2) < \min\{\ep_{N_0 + 1}(\bxi_1), \ep_{N_0 + 1}(\bxi_2)\}$
for all $(\bxi_1, \bxi_2) \in (\Gamma^*)^2$.
Consequently, the analyticity of $\widehat{G}$ depends on the regularity of the mapping $\nu: (\Gamma^*)^2 \to \C$. Analogously to the proof of \Cref{thm:nearsight}, our goal is to construct a continuous shift $\nu$ and identify its domain of analyticity.

First, as in the proof of \Cref{thm:nearsight}, we define
\begin{align} 
    \label{eq: sqrt_D_R}
    D_{\rm R} 
        \coloneqq
            \left\{ 
                \bxi + i\bzeta \colon \bxi \in \Gamma^*, |\bzeta| \leq \min\Big\{\tfrac{1}{2}\gamma_0,\tfrac{1}{2}\reg, \sqrt{\tfrac{\alpha}{m_{\infty}} \gap} \Big\}
            \right\}
\end{align}
with $\alpha := \frac14$. For ${\bm \omega_l} \in D_{\rm R}$ for $l=1,2$, we define
\begin{align*}
    \nu_0({\bm \omega_1} , {\bm \omega_2})                      \coloneqq 
            \frac{\min_{\bxi_1, \bxi_2}\ep_{N_0 + 1} + \max_{\bxi_1, \bxi_2}\ep_{N_0}}{2} \in \R
\end{align*}
and construct a polynomial approximation $\nu$ satisfying 
\begin{align}
    |\re \nu({\bm \omega_1} , {\bm \omega_2}) - \nu_0({\bm \omega_1} , {\bm \omega_2}) | \leq \frac{\alpha}{2}\gap,
    \quad
    ({\bm \omega_1} , {\bm \omega_2}) \in (D_{\rm R})^2.
\end{align}
In the exact same way as in \eqref{eq:distz_nu}, we have 
\begin{align*}
    {\rm dist}\left(\re\nu({\bm \omega}_1, {\bm \omega}_2), \re \sigma(\Ham_{{\bm \omega}_1}) \cup \re \sigma(\Ham_{{\bm \omega}_2}) \right)
        \geq
            \frac{1}{4}\gap,
\end{align*}
for all ${\bm \omega}_1, \bm \omega_2 \in D_{\rm R}$. In particular, $\left( \cc + \nu({\bm \omega}_1, {\bm \omega}_2)\right) \cap [\sigma(\Ham_{{\bm \omega}_1}) \cup \sigma(\Ham_{{\bm \omega}_2})] = \emptyset $ for all $\bm \omega_1, \bm \omega_2 \in D_{\rm R}$.
As a result, the function 
 $\widehat{G}$ \eqref{eq: res_time_res} is analytic on $(D_{\rm R})^2$. 

Let us define the decay rate $\widetilde{\eta}_- \coloneqq \min\big\{\tfrac{1}{2}\gamma_0, \tfrac{1}{2}\reg, \sqrt{\tfrac{\alpha}{m_{\infty}} \gap}\big\}$. Then, for any ${\bm \omega}_l = \bm \xi_l + i \bm \zeta_l$ with  $|\bzeta_l|\leq \widetilde{\eta}_-$, $l = 1,2$, we have
\begin{align*}
    \tilde{C}^2\coloneqq \int_{(\Gamma^*)^2} \big|\widehat{G}(\bm \omega_1,  \bm \omega_2)_{\ell_0 m_0, n_0\ell_0}^{aa', bb'}\big|^2 \dd \bxi_1 \dd \bxi_2
        = \sum_{mn} \left|G_{\ell m, n\ell}^{aa', bb'}\right|^2 e^{2\,\bzeta_1 \cdot (\ell - m) + 2 \,\bzeta_2 \cdot (n - \ell)}  < \infty,
\end{align*}
Note that $\big|\widehat{G}(\bm \omega_1,  \bm \omega_2)_{\ell_0 m_0, n_0\ell_0}^{aa', bb'}\big| \leq (\gap)^{-2}$ for any $(\bm \omega_1, \bm \omega_2) \in (D_{\rm R})^2$, so 
$|G_{\ell m, n\ell}^{aa', bb'}| \leq \frac{\tilde{C}}{\gap^2} e^{-\widetilde{\eta}_-(r_{\ell m} + r_{n \ell})}$.
Combining this with \eqref{eq: deriv_contour_integral}, we obtain 
\begin{align*}
    \left|
           \frac{\partial\rho_{\ell\ell, ab}}{\partial [\bm r_k]_{j}}
    \right| 
            &\leq
            \frac{1}{2\pi}\ell(\cc)h_0(M\nob)^2 \frac{\tilde{C}}{\gap^2}
            \sum_{mn}
            \left|\frac
                {\partial \Ham_{mn}^{a'b'}}
                {\partial [\bm r_k]_{j}}\right| 
            \left|G_{\ell m, n\ell}^{aa', bb'}\right|
            \nonumber\\
            &\leq 
            \frac{1}{2\pi}\ell(\cc)h_0(M\nob)^2 \frac{\tilde{C}}{\gap^2}
            \sum_{mn}
            e^{-\gamma_0 (r_{mk} + r_{kn})}e^{-\widetilde{\eta}_- (r_{\ell m} + r_{\ell n})}
            \nonumber\\
            &\leq
            \frac{\tilde{C}}{\gap^2} e^{-\widetilde{\eta}_-r_{\ell k}},
\end{align*}
where the constant $\tilde{C} > 0$ depends on $h_0, \gamma_0, d, \m, M$, and $\nob$. 

For the nearsightedness estimate \eqref{eq:shaper_nearsight}, define $\widetilde{\eta}_+ \coloneqq \min\big\{\tfrac{1}{2}\gamma_0, \tfrac{1}{2}\reg, \sqrt{\tfrac{\alpha}{m_{\infty}} \Gap}\big\}$. Together with the perturbation estimate \eqref{eq: real_ptb}, the same argument as in the proof of Theorem~\ref{thm:nearsight} (with $\etaP$ replaced by $\widetilde{\eta}_+$), shows that \eqref{eq:distz_nu} holds for all $|\bxi| < \widetilde{\eta}_+$ and hence \eqref{eq:shaper_nearsight} follows. 

\appendix
\section{Bloch Transform for the Tight Binding Hamiltonian}
\label{appd:bloch_trans}
Recall that $\Gamma, \Gamma^*$ are unit cells of $\Lambda$ and the corresponding reciprocal lattice, respectively. Let us denote all the lattice sites in the unit cell by $\Lambda_0 = \Gamma \cap \Lambda$, and denote the set of orbitals by $\Xi$.

Given an atom site $k \in \Lambda$, there exits a vector $\alpha_{k} \in \Z^d$ such that $k = k_0 + \A\alpha_{k}$, $k_0 \in \Lambda_0$. Then, we define the operator $U: \ell^2(\Lambda \times \Xi) \rightarrow L^2\left(\Gamma^*, \ell^2(\Lambda_0 \times \Xi)\right)$ as
\begin{align*}
    \begin{split}
        (U\psi)_{\bxi}(k; a) := \hat{\psi}_{\bxi}(k_0;a)=\sum_{\alpha \in \mathbb Z^d}\psi\big(k_0 + \A\alpha; a\big)e^{-i\bxi \cdot (k_0 + \A\alpha)}, 
    \end{split} 
\end{align*}
where $\bxi \in \Gamma^*$, $k_0 \in \Lambda_0$, and $1 \leq a \leq \nob$. The space $L^2\left(\Gamma^*, \ell^2(\Lambda_0 \times \Xi)\right)$ is a Hilbert space endowed with the inner product
\begin{align*}
\langle \hat{\psi}, \hat{\phi}\rangle_{L^2\left(\Gamma^*, \ell^2(\Lambda_0 \times \Xi)\right)} \coloneqq \dashint_{\Gamma^*} \langle \hat{\psi}_{\bxi}(\cdot), \hat{\phi}_{\bxi}(\cdot)\rangle_{\ell^2(\Lambda_0 \times \Xi)}\dd\bxi = \sum_{\substack{k_0 \in \Lambda_0, a \in \Xi}}\dashint_{\Gamma^*}\hat{\psi}^*_{\bxi} (k_0;a)\hat{\phi}_{\bxi}(k_0;a)\dd\bxi.
\end{align*}
It is straightforward to verify that $U$ is an isometry with inverse $U^{*}:L^2\left(\Gamma^*, \ell^2(\Lambda_0 \times \Xi)\right) \rightarrow \ell^2(\Lambda \times \Xi)$ defined by 
\begin{align*}
    \begin{split}
        (U^{*}\hat{\psi}_{\bxi})(k_0 + \A\alpha):= \dashint_{\Gamma^*}\hat{\psi}_{\bxi}(k_0;a)e^{i\bxi\cdot (k_0 + \A\alpha)}\dd \bxi,
    \end{split}
\end{align*}
for $k_0 \in \Lambda_0$ and $\alpha\in\mathbb Z^d$.

The reference Hamiltonian $\Hr$ we defined in \eqref{eq:ham} is a linear self-adjoint operator on $\ell^2(\Lambda \times \Xi)$. Following \cite{ weinan2010electronic, reed1978iv}, there exits a direct integral decomposition of $\Hr$ given by
\begin{align*}
    U \Hr U^{*} = \dashint_{\Gamma^*}^{\oplus} \Hxi \dd\bxi, 
\end{align*}
where $\dashint_{\Gamma^*}^{\oplus}$ denotes the direct integral and $\Hxi$ satisfies
$(U\Hr\psi)_{\bxi} = \Hxi\hat{\psi}_{\bxi}$
 for almost every $\bxi \in \Gamma^*$. 
Besides, we have
\begin{align*}
    [\Hxi]_{\ell_0 k_0}^{ab} \coloneqq \sum_{\alpha\in \mathbb Z^d} [\Hr]^{ab}_{\ell_0 + \mathsf{A}\alpha, k_0} e^{- i\bxi \cdot (\ell_0 - k_0 +\mathsf{A}\alpha) },
    \quad
    \forall \bxi \in \Gamma^*, \,\ell_0, k_0 \in \Lambda_0,
\end{align*}
where $\Hxi$ is a $\na\nob\times \na\nob$ matrix.

For any $\bxi \in \Gamma^*$, by solving the eigenvalue problem
\begin{align}
\label{eq:bloch_eigenpair}
   \Hxi u_{n}(\bxi) = \ep_{n}(\bxi)u_{n}(\bxi), \quad n = 1, \dots, \na\nob, 
\end{align}
we get the spectrum $\{\ep_{n}(\bxi)\}_{n=1}^{M\nob}$, where the ordered eigenvalues $\ep_{1}(\bxi) \leq \cdots \leq \ep_{M\nob}(\bxi)$ are continuous on $\Gamma^*$ \cite{des1964analytical, kohn1959analytic}. Finally, we obtain the band structure:
\begin{align*}
     \sigma(\Hr) = \bigcup_{j=1}^{M\nob}\left\{\ep_{n}(\bxi):  \bxi \in \Gamma^*\right\}.
\end{align*}

\section{Computational Details in 1D Chain}
\label{appd: 1d_chain}
Here we will provide the computational details of 1D chain tests.
The Bloch transform gives that
\begin{align*}
    \Ham_{\xi} = 
        {\bm c} + 2{\bm f}(1)\cos\xi,
        \quad
        \forall\xi\in \Gamma^* = [-\pi, \pi].
\end{align*}
We solve the eigenvalue problem of $\Ham_{\xi}$ by considering the equation 
$
    \det\big|\ep(\xi) - \Ham_{\xi}\big| = 0,
$
which yields two energy bands denoted by $\ep_\pm(\xi)$.

In particular, for the model of type (i) where $c_{12} = c_{21} = 0$, we have
\begin{align*}
\begin{split}
    \ep_\pm(\xi)
    &=\frac{c_{11}+c_{22}}{2} + \big(f_{11} + f_{22}\big) \cos\xi \pm\sqrt{\frac{\big(c_{11}-c_{22} + 2(f_{11} - f_{22} \cos\xi\big)^2}{4}+4\left(f_{12} \cos \xi\right)^2},
\end{split}
\end{align*}
where we denote $f_{ij} \coloneqq f_{ij}(1)$ for $1 \leq i, j \leq 2$.
Then $\Gap$ can be given immediately by
\begin{align}
\label{eq:1d_Gap}
    \mathsf{gap}_{+} 
    &= \min_{\xi \in \Gamma^*} \big(\ep_+(\xi) - \ep_-(\xi)\big) 
     = \min_{\xi \in \Gamma^*} \sqrt{{\big(c_{11}-c_{22} + 2(f_{11} - f_{22} \cos\xi)\big)^2}+16\left(f_{12} \cos \xi\right)^2} \nonumber\\
    &= \frac{2 f_{12}\left|c_{11}-c_{22}\right|}{\sqrt{\big(f_{11}-f_{22}\big)^2+ 4 f_{12}^2}}.
\end{align}
Additionally, for the model of type (ii) where $f_{12} = f_{21} = 0$, the associated energy bands are given by
\begin{align*}
\begin{split}
    \ep_\pm(\xi)
    =\frac{c_{11}+c_{22}}{2} + \big(f_{11} + f_{22}\big) \cos\xi \pm\sqrt{\frac{\big(c_{11}-c_{22} + 2(f_{11} - f_{22} \cos\xi)\big)^2}{4}+4c_{12}^2},
\end{split}
\end{align*}
where the $\Gap$ is calculated by 
\begin{align*}
    \Gap 
     = \min_{\xi \in \Gamma^*} \big(\ep_+(\xi) - \ep_-(\xi)\big)
    = \sqrt{\big(c_{11} - c_{22} \pm 2 (f_{11} - f_{22})\big)^2 + 16 c_{12}^2}.
\end{align*}

Next, we describe the parameters choices in our numerical experiments. For the type (i) model, we set $c_{11} = 1 + 1/2\Gap, c_{22} = 1 - 1/2\Gap$ and
\begin{align*}
    f_{11} &= (0.5 + 0.5 r)\eta(r),~
    f_{22} = (0.3 + 0.7 r)\eta(r), \\
    f_{13} &= (0.4 + (p - 0.4)r)\eta(r),
\end{align*}
where $p = \frac{1}{2}\sqrt{(\frac{1}{2}\gap + 2)^2 - \frac{1}{4}\Gap}$ and $\eta(r) = e^{1-1/(1 - |r-1|^2)}$ as $r \in (1, 2)$ which smoothly connects $\eta = 1$ for $r \in [0, 1]$ and $\eta = 0$ for $r \in [2, \infty)$. Under this parametrization, the band functions take the form
\begin{align*}
    \ep_{\pm}(\xi)
        = (1 + 2\cos\xi) \pm \frac{1
        }{2}\sqrt{\Gap^2 + [(\gap + 4)^2 - \Gap^2]\cos^2\xi }.
\end{align*}
This allows us to design experiments with either fixed $\gap$ or fixed $\Gap$. Moreover, as $\Gap$ and $\gap$ tend to zero, we obtain the limiting band behaviour
$
\ep_\pm(\xi) \approx 1 + 2(\cos\xi \pm |\cos\xi|)
$
which corresponds to the case shown in  \Cref{fig:1d_bands_type11} .

For the second type, we set $c_{11} = -1, c_{22} = 5, c_{12} = c_{21} = 1/2\Gap, f_{11}(r) = (1 + r)\eta(r)$ and $f_{22}(r) = (-2 + r)\eta(r)$, which yields cosine-like bands 
\begin{align*}
    \ep_{\pm}(\xi)
        = (2 + \cos\xi) \pm \frac{1
        }{2}\sqrt{\big(-6 +  6\cos\xi \big)^2 + \Gap^2}.
\end{align*}
As $\Gap \to 0 $, we observe that
$\ep_+(\xi) \approx -1 + 4\cos\xi$, and $\ep_-(\xi) \approx 5 - 2\cos\xi$
which corresponds to the behavior illustrated in \Cref{fig:1d_bands_type2}.

\begin{remark}
    For the type (i), an alternative definition of  $\Gap$ could be used:
    \[
    \Gap = \big[c_{11} - c_{22} \pm 2 (f_{11} - f_{22}) \big]^2 + 16 f_{12}^2,
    \]
    since the minima of $\ep_+(\xi) - \ep_-(\xi)$ occur at $0$ or $\pm \pi$.
    However, in our tests, the chosen parameters yield $\Gap$ as defined by \eqref{eq:1d_Gap}.
\end{remark}

\section{Exponential Sums}
Throughout \Cref{sec:Proofs of the Main Results}, we frequently apply estimates for exponential sums. For the convenience of the reader, we present these elementary estimates here.

\begin{lemma}
\label{lem:expsum}
For $\gamma \geq \eta > 0$, $\ell, k \in \Lambda$, and $A \subseteq \Lambda$, 
\begin{align}
    \sum_{m \in A} 
    e^{-\gamma r_{\ell m}} e^{-\eta r_{mk}}
    &\leq
    \left[ 1 + \frac
        {C_d}
        {(\m\gamma)^d}
    \right]
    e^{ -\frac{1}{2} \min\limits_{m \in A}[\gamma r_{\ell m} + \eta r_{mk} ]}  \label{eq:expsum-1}\\
    &\leq 
    \left[ 1 + \frac
        {C_d}
        {(\m\gamma)^d}
    \right]
    e^{ \gamma \sup\limits_{a\in A} |a| }
    e^{ -\frac{1}{2} (\gamma |\ell| + \eta |k| )}
    \label{eq:expsum-2}
\end{align}
for some $C_d$ depending only on $d$.
\end{lemma}

We apply Lemma~\ref{lem:expsum} with $A = \Lambda$ and replace the exponent in the first line with $\eta r_{\ell k} \leq \min\limits_{m\in \Lambda}[\gamma r_{\ell m} + \eta r_{mk}]$. When $A = \Lambda \cap B_R$, we may apply the second line with $\sup\limits_{a\in A}|a| = R$.

\begin{proof} Here, we follow \cite{thomas2021analysis}. Since $\gamma r_{\ell m } + \eta r_{mk}\geq \min\limits_{m \in A}[\gamma r_{\ell m} + \eta r_{mk} ]$ and $\gamma r_{\ell m } + \eta r_{mk}\geq \gamma r_{\ell m}$, we have
\begin{align}
\label{eq:exp_sum_1}
    \sum_{m \in A} e^{-\gamma r_{\ell m}} e^{-\eta r_{mk}} 
        &\leq 
            \Big(
                \sum_{m \in A} e^{-\frac{1}{2}\gamma r_{\ell m}}
            \Big) e^{ -\frac{1}{2} \min\limits_{m \in A}[\gamma r_{\ell m} + \eta r_{mk} ]} .
\end{align}
Additionally, by noting that $e^{-\gamma r_{\ell m }} \leq e^{-\gamma|\br - \br_{\ell}|}$ for all $\br \in B_{r_{\ell m}}(\br_{\ell})$, and the non-interpenetration condition in $\eqref{eq:adm_configs}$, 
we can approximate \cref{eq:exp_sum_1} with the following integral
\begin{align*}
\label{eq:int_sum}
    \sum_{m \neq \ell} e^{-\frac{1}{2}\gamma r_{\ell m}}
    &\leq
        \sum_{m \neq \ell} \dashint_{B_{r_{\ell m }} (\br_{\ell}) \cap B_{\m/2}(\br_{m})} e^{-\frac{1}{2}\gamma |\br - \br_{\ell}|}\dd\br
    \leq
        \frac{C_d}{\m^d}
        \int_{\R^d} e^{-\frac{1}{2}\gamma |\br - \br_{\ell}|}\dd\br
        \\
    &=\frac{C_d}{\m^d}\int_0^{\infty} e^{-\frac{1}{2}\gamma r} r^{d-1}\dd r
        = \frac{C_d}{(\m\gamma)^d},
\end{align*}
where the constant $C_d$ changes from one line to the next. Together with \eqref{eq:exp_sum_1}, this concludes the proof of \eqref{eq:expsum-1}.
We conclude \eqref{eq:expsum-2} by noting that 
\begin{align*}
    \frac{1}{2} \min\limits_{m \in A}[\gamma r_{\ell m} + \eta r_{mk} ]
    &\geq \frac{1}{2} \big[\gamma \dist(\ell, A) + \eta \dist(k, A)\big] \nonumber\\
    &\geq  \frac{1}{2} \big(\gamma |\ell| + \eta |k|\big)
    - \gamma \sup_{a\in A} |a|.
\end{align*}
\end{proof}

\bibliography{refs.bib}
\bibliographystyle{siamplain}

\end{document}